\documentclass{lmcs} 
\pdfoutput=1

\usepackage{lastpage}
\lmcsdoi{16}{4}{5}
\lmcsheading{}{\pageref{LastPage}}{}{}%
{Oct.~11,~2019}{Oct.~16,~2020}{}

\usepackage[utf8]{inputenc}
\usepackage{amssymb}
\usepackage{amsmath}
\usepackage{graphicx}
\usepackage{accents}
\usepackage{xcolor}
\usepackage{xargs}
\usepackage{mathtools}

\usepackage{proof}
\usepackage{alltt}
\usepackage{url}
\usepackage{stmaryrd}
\usepackage{scalerel}
\usepackage{subfigure}
\usepackage{hyperref}
\usepackage{cleveref}

\usepackage[all]{xy}\CompileMatrices


\theoremstyle{plain} 

\newcommand{\reach}[1]{{\it reach}(#1)}

\newcommand{\voidt}{\epsilon}
\newcommand{\trace}{\omega}
\newcommand{\fseq}{s}

\newcommand{\ceq}{\asymp}

\newcommand{\len}[1]{|#1|}

\newcommand{\m}{Mazurkiewicz }

\newcommand{\supp}{\mathit{supp}}

\newcommand{\minp}[1]{\ensuremath{{}^\circ{#1}}}
\newcommand{\maxp}[1]{\ensuremath{#1^\circ}}

\newcommand{\preS}[1]{\ensuremath{{}^\bullet{#1}}}
\newcommand{\postS}[1]{\ensuremath{#1^\bullet}}

\newcommand{\places}{\mathcal{P}}
\newcommand{\p}[1]{{\sf #1}}

\newcommand{\transitions}{\mathcal{T}}
\newcommand{\tr}[1]{{\sf #1}}

\newcommand{\nat}{\mathbb{N}} 

\newcommand{\multiset}[1]{#1}

\newcommandx{\red}[1][1=\tr t]{\xrightarrow{#1}}
\makeatletter
\newcommand{\xMapsto}[2][]{\ext@arrow 0599{\Mapstofill@}{#1}{#2}}
\def\Mapstofill@{\arrowfill@{\Mapstochar\Relbar}\Relbar\Rightarrow}
\makeatother

\newcommandx{\fred}[1][1=\tr t]{\stackrel{#1}{\twoheadrightarrow}}

\newcommandx{\bred}[1][1=\tr t]{\stackrel{#1}{\rightsquigarrow}}

\newcommand{\pntrans}{[\rangle}

\def \mathrule #1#2#3{\begin{array}{l}%
    {\mbox{\scriptsize ({\sc #1})} }%
    \\ \irule{#2}{#3}%
\end{array}}
\newcommand{\irule}[2]{\frac{\textstyle\rule[-1.3ex]{0cm}{3ex}#1}%
{\textstyle\rule[-.5ex]{0cm}{3ex}#2}}

\makeatletter
\DeclareRobustCommand{\cev}[1]{%
  \mathpalette\do@cev{#1}%
}
\newcommand{\do@cev}[2]{%
  \fix@cev{#1}{+}%
  \reflectbox{$\m@th#1\overrightarrow{\reflectbox{$\fix@cev{#1}{-}\m@th#1#2\fix@cev{#1}{+}$}}$}%
  \fix@cev{#1}{-}%
}
\newcommand{\fix@cev}[2]{%
  \ifx#1\displaystyle
    \mkern#23mu
  \else
    \ifx#1\textstyle
      \mkern#23mu
    \else
      \ifx#1\scriptstyle
        \mkern#22mu
      \else
        \mkern#22mu
      \fi
    \fi
  \fi
}

\makeatother

\newcommand{\rev}[1]{\cev{#1}}
\newcommand{\enc}[1]{\llbracket #1\rrbracket}

\definecolor{mygray}{RGB}{120,120,120}

\newcommand{\drawplace}{*[o]=<1.2pc,1.2pc>{\ }\drop\cir{}}
\newcommand{\drawmarkedplace}{*[o]=<1.2pc,1.2pc>{\bullet}\drop\cir{}}
\newcommand{\drawcolouredmarkedplace}[1]{*[o]=<1.2pc,1.2pc>{#1}\drop\cir{}}

\newcommand{\nameplaceup}[1]{\POS[]+<0pc,1pc>\drop{{#1}}}
\newcommand{\nameplacedown}[1]{\POS[]-<0pc,1pc>\drop{{#1}}}
\newcommand{\nameplaceright}[1]{\POS[]+<1pc,0pc>\drop{{#1}}}
\newcommand{\nameplaceleft}[1]{\POS[]-<1pc,0pc>\drop{{#1}}}

\newcommand{\drawtrans}[1]{*=<2pc,1.3pc>{{#1}}\drop\frm{-}}

\newcommand{\colourset}{\mathcal{C}}

\newcommand{\varset}{\mathcal{X}}
\newcommand{\vars}[1]{{\it vars}(#1)}

\begin{document}

\title[Reversing P/T Nets]{Reversing Place Transition Nets}
\titlecomment{{\lsuper*}The authors would like to thank for partial
  support by COST Action IC1405 on Reversible Computation - Extending
  Horizons of Computing. This paper is a revised and extended version
  of \cite{MelgrattiMU19}.}

\author[H.~Melgratti]{Hern\'an Melgratti}
\address{Instituto de Ciencias de la Computaci\'on - Universidad de
  Buenos Aires - Conicet, Argentina } 
\email{hmelgra@dc.uba.ar} 
\thanks{The first author is partially supported by the EU H2020 RISE
  programme under the Marie Sk\l{}odowska-Curie grant agreement
  778233, by the UBACyT projects 20020170100544BA and
  20020170100086BA, and by the PIP project 11220130100148CO} 

\author[C.A.~Mezzina]{Claudio Antares Mezzina}	
\address{Dipartimento di Scienze Pure e Applicate, Universit\`a di
  Urbino, Italy } 
\email{claudio.mezzina@uniurb.it} 
\thanks{The second author has been partially supported by French ANR
  project DCore ANR-18-CE25-0007 and by the Italian INdAM -- GNCS
  project 2020 Reversible Concurrent Systems: from Models to
  Languages.} 

\author[I.~Ulidowski]{Irek Ulidowski} 

\address{School of Informatics, University of Leicester, England} 

\email{iu3@leicester.ac.uk}  

\newtheorem{innercustomlem}{Lemma}
\newenvironment{lm}[1]
  {\renewcommand\theinnercustomlem{#1}\innercustomlem}
  {\endinnercustomlem}
\newtheorem{innercustomtm}{Theorem}
\newenvironment{tm}[1]
  {\renewcommand\theinnercustomlem{#1}\innercustomtm}
  {\endinnercustomlem}




\begin{abstract}
\noindent Petri nets are a well-known model of concurrency and provide
an ideal setting for the study of fundamental aspects in concurrent
systems.  Despite their simplicity, they still lack a satisfactory
causally reversible semantics.  We develop such semantics for
Place/Transitions Petri nets (P/T nets) based on two observations.
Firstly, a net that explicitly expresses causality and conflict among
events, for example an occurrence net, can be straightforwardly
reversed by adding a reverse transition for each of its forward
transitions.  Secondly, given a P/T net the standard unfolding
construction associates with it an occurrence net that preserves all
of its computation. Consequently, the reversible semantics of a P/T
net can be obtained as the reversible semantics of its unfolding.
We show that such reversible behaviour can be expressed as a finite
net whose tokens are coloured by causal histories.  Colours in our
encoding resemble the causal memories that are typical in reversible
process calculi.
\end{abstract}

\maketitle

\section{Introduction}

Reversible computing is attracting interest for its applications in
many fields including hardware design and quantum
computing~\cite{2018Vos}, the modelling of bio-chemical
reactions~\cite{CardelliL11,PhiUliYuen12,UK16,Pinna17,UK18,KACPPU20},
parallel discrete event simulation~\cite{SchordanOJB18} and program
reversing for
debugging~\cite{GiachinoLM14,Lanese0PV18,JH2018,HoeyLNUV20}.

A model for reversible computation features two computation flows: the
standard forward direction and the reverse direction, which allows us
to go back to any past state of the computation.
Reversibility is well understood in a sequential setting in which
executions are totally ordered sets of events (see~\cite{Leeman}): a
sequential computation can be reversed by successively undoing the
last not yet undone event.
Reversibility becomes more challenging in a concurrent setting because
there is no natural way to define a total ordering on concurrent
events.
Many concurrency models represent the causal dependencies among
concurrent events using partial orders.
Reversing an execution of a partially ordered set of events then
reduces to successively undoing one of the maximal events which has
not been undone yet.  This is the basis of
the \textit{causal-consistent
reversibility}~\cite{rccs,ccsk,cc_bulletin,revaxiom}, which relates
reversibility with concurrency and causality.
Intuitively, this notion stipulates that any event can be undone
provided that all its consequences, if any, are undone beforehand.
Reversibility in distributed systems, for example in
checkpoint/rollback protocols~\cite{VassorS18} and in
transactions~\cite{DanosK05,LaneseLMSS13,stm}, can also be modelled by
causal-consistent reversibility.
The interplay between reversibility and concurrency and causality has
been widely studied in process
calculi~\cite{rccs,ccsk,rhotcs,CristescuKV13,family_pi}, event
structures~\cite{PhillipsU15,CristescuKV15,UlidowskiPY18,GraversenPY18},
and lately in Petri
nets~\cite{BarylskaKMP18,PhilippouP18,BarylskaGMPPP18}.
Despite being a very basic model of concurrency, previous works on
Petri nets only provide causally-consistent reversible semantics for
subclases of Petri nets (e.g., acyclic nets or
backward-concurrent-free).

A key point when reversing computations in Petri nets is to handle
backward conflicts, namely that a token can be generated in a place
due to different causes.
Consider the net in Figure \ref{fig:back_confl} showing the initial
state of a system that can either perform $\tr t_1$ followed by $\tr
t_3$, or $\tr t_2$ followed by $\tr t_4$.
The final state of a complete computation is depicted in
Figure \ref{fig:back_confl_final}.
The information in that state is not enough to deduce whether the
token in $\p d$ has been produced because of $\tr t_3$ or $\tr t_4$.
Even worse, if we ``naively'' reverse the net by just adding
transitions in the reverse direction, as shown in
Figure \ref{fig:back_confl_rev} (where reverse firings are represented
by dashed arrows and the names of their transitions are decorated with
`$\rev{}$' as in, for example, $\rev{\tr{t_1}}$), the reverse
transitions will do more than just undoing the computation.
In fact, the token in $\p d$ can be put back either in $\p b$ or in
$\p c$ regardless of which place it came from.

\begin{figure}[t]
  \subfigure[Initial]{	
  \label{fig:back_confl}
  $$
  \xymatrix@R=1pc@C=.2pc{
    &
    \drawmarkedplace\ar[dl]\ar[dr]
    \nameplaceup {\p a}
    \\           
    \drawtrans{\tr{t_1}} \ar[d]
    &&
    \drawtrans{\tr{t_2}} \ar[d]
    \\
    \drawplace\ar[d]
    \nameplaceleft {\p b}
    &
    &
    \drawplace\ar[d]
    \nameplaceleft {\p c}
    \\
    \drawtrans{\tr{t_3}} \ar[dr]
    &&
    \drawtrans{\tr{t_4}} \ar[dl]
    \\
    &
    \drawplace
    \nameplacedown {\p d}
    }
  $$
  }
  \hspace{1cm}
  \subfigure[Final]{	
  \label{fig:back_confl_final}
  $$
  \xymatrix@R=1pc@C=.2pc{
    &
    \drawplace\ar[dl]\ar[dr]
    \nameplaceup {\p a}
    \\
    \drawtrans{\tr{t_1}} \ar[d]
    &&
    \drawtrans{\tr{t_2}} \ar[d]
    \\
    \drawplace\ar[d]
    \nameplaceleft {\p b}
    &&
    \drawplace\ar[d]
    \nameplaceleft {\p c}
    \\
    \drawtrans{\tr{t_3}} \ar[dr]
    &&
    \drawtrans{\tr{t_4}} \ar[dl]
    \\
    &
    \drawmarkedplace
    \nameplacedown {\p d}
    }
    $$
  }
  \hspace{1cm}
  \subfigure[Naive reversing]
  {
  \label{fig:back_confl_rev}
  $$
  \xymatrix@R=1pc@C=.2pc{
    &&
    \drawplace\ar[dl]\ar[dr]
    \nameplaceup{\p a}
    &
    &
    \\
    \drawtrans{ \rev{\tr{t_1}}} \ar@{-->}@/^/[urr]	
    &
    \drawtrans{\tr{t_1}} \ar[d]
    &&
    \drawtrans{\tr{t_2}} \ar[d] 
    &
    \drawtrans{\rev{\tr{t_2}}} \ar@{-->}@/_/[ull]
    \\
    &
    \drawplace\ar[d]\ar@{-->}[ul]
    \nameplaceleft {\p b}
    &&
    \drawplace\ar[d]\ar@{-->}[ur]
    \nameplaceleft {\p c}
    \\
    \drawtrans{ \rev{\tr{t_3}}} \ar@{-->}[ur]	
    &
    \drawtrans{\tr{t_3}} \ar[dr]
    &&
    \drawtrans{\tr{t_4}} \ar[dl] 
    & \drawtrans{\rev{\tr{t_4}}} \ar@{-->}[ul]
    \\
    &&
    \drawmarkedplace{\p d} \ar@{-->}@/^/[ull] \ar@{-->}@/_/[urr]
    \nameplacedown {\p d}
  }
  $$
}
\label{fig:init}
\caption{Backward conflict and naive reversing.}
\end{figure}
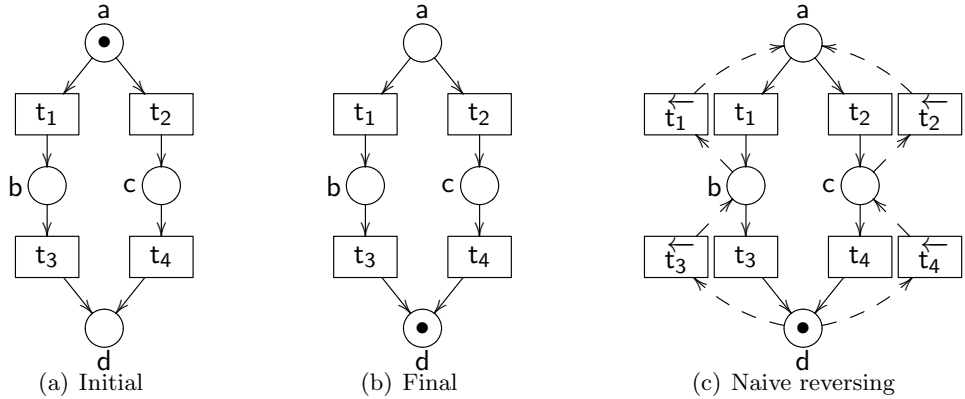
Analogous problems arise when a net is cyclic.
Previous approaches~\cite{PhilippouP18,BarylskaGMPPP18} to reversing
Petri nets tackle backward conflicts by relying on a new kind of
tokens, called {\em bonds}, that keep track of the execution history.
Bonds are rich enough for allowing other approaches to reversibility
such as \textit{out-of-causal order}
reversibility~\cite{PhiUliYuen12,UK16,UK18}, however they do
not apply to general nets but only to nets that are either
acyclic~\cite{PhilippouP18,BarylskaGMPPP18} or
backward-concurrent-free~\cite{PhilippouPS19}.

We base our paper on the following simple question:
\begin{quote}\begin{center}
\emph{Can we define a reversible semantics for Petri nets by relying on standard notions of Petri net theory?}
\end{center}\end{quote}
We provide here an affirmative answer to this question.  We propose a
reversible model for \textsc{p/t} nets that can handle cyclic nets by
relying on standard notions in Petri net theory.
We first observe that a Petri net can be mapped via the standard
unfolding construction to an occurrence net, i.e., an acyclic net that
does not have backward conflicts and makes causal dependencies
explicit.
Then, an occurrence net can be reversed by adding a reverse transition
for each of its transitions. Such construction gives a model where
causal-consistent reversibility holds (shown in Section~4).
We also prove that each reachable marking in the reversible version of
an occurrence net can be reached by just forward computational
steps. We observe that the unfolding construction could produce an
infinite occurrence net.
However, the unfolding can be seen as the definition of a
coloured net, where colours account for causal histories. Such
interpretation associates a \textsc{p/t} net with an equivalent
coloured \textsc{p/t} net, which can be reversed correspondingly as
described above.
The correctness of the construction is shown by exhibiting a
one-to-one correspondence of its executions with the ones of the
reversible version of the unfolding.  Interestingly, the colours used
by the construction resemble the memories common in reversible
calculi~\cite{rccs,LaneseMSS11,rhotcs}.

We remark that our proposal deals with reversing (undoing)
computations in a Petri net and not with the classical problem of
reversibility~\cite{revPT} which requires every computation to be able
to reach back the initial state of the system (but not necessary by
undoing the previous events).
In the latter case, the problem of making a net reversible equates to
adding a minimal amount of transitions that make a net
reversible~\cite{BarylskaKMP18,MikulskiL19}.
Reversibility is a global property while reversing a computation is a
local one, as discussed in~\cite{BarylskaKMP18,MikulskiL19}.

This paper is an extended version of \cite{MelgrattiMU19}, which
includes proofs of our results and many additional examples and
explanations.

\section{Background}\label{sec:bg}
\subsection{Petri Nets}\label{sec:pt}

\noindent Petri nets are built up from \emph{places} (denoting, e.g.,
resources and message types), which are repositories for {\em tokens}
(representing instances of resources), and \emph{transitions}, which
fetch and produce tokens.
We consider the infinite sets $\places$ of places and $\transitions$
of transitions, and assume that they are disjoint, i.e., $\places \cap
\transitions = \emptyset$. We let $\p{a}, \p{a}', \ldots$ range over
$\places$ and $\tr{t}, \tr{t'}, \ldots$ over $\transitions$. We write
$x, y, \ldots$ for elements in $\places \cup \transitions$.
 
A \emph{multiset} over a set $S$ is a function $m: S \rightarrow \nat$
(where $\nat$ denotes the natural numbers including zero).
We write $\nat^{S}$ for the set of multisets over $S$.
For $m\in\nat^{S}$, $\supp(m)= \{x \in S \; | \; m(x) > 0 \}$ is the
{\em support} of $m$, and $|m|= \sum_{x\in S} m(x)$ stands for its
{\em cardinality}.  We write $\emptyset$ for the empty multiset, i.e.,
$supp(\emptyset) = \emptyset$.
The union of $m_1,m_2\in\nat^{S}$, written $(m_1 \oplus m_2)$, is
defined such that $(m_1 \oplus m_2)(x) = m_1(x)+m_2(x)$ for all $x\in
S$.
Note that $\oplus$ is associative and commutative, and has $\emptyset$
as identity. Hence, $\nat^S$ is the free commutative monoid $S^\oplus$
over $S$. We write $x$ for a singleton multiset, i.e., $\supp(x)=
\{x\}$ and $m(x)=1$.
Moreover, we write $\multiset{x_1\ldots x_n}$ for
$x_1\oplus\ldots\oplus x_n$.  Let $f : S \rightarrow S'$, we write $f$
also for its obvious extension to multisets, i.e.,
$f(\multiset{x_0\ldots x_n}) = \multiset{f(x_0)\ldots f(x_n)}$. We
avoid writing $\supp(\_)$ when applying set operators to multisets,
e.g., we write $x \in m$ or $m_1\cap m_2$ instead of $x\in \supp(m)$
or $\supp(m_1)\cap \supp(m_2)$.

\begin{defi}[Petri Net]\label{def:net}
  A \emph{net} $N$ is a 4-tuple $N = (S_N, T_N, \preS{\_}_N,
  \postS{\_}_N)$ where $S_N\subseteq \mathcal{P}$ is the (nonempty)
  set of places, $T_N \subseteq \transitions$ is the set of
  transitions and the functions $\preS{\_}_N, \postS{\_}_N:
  T_N\rightarrow 2^{S_N}$ assign source and target to each transition
  such that $\preS{\tr{t}} \neq\emptyset$ and $\postS{\tr{t}}
  \neq\emptyset$ for all $\tr{t}\in T_N$.
  A marking of a net $N$ is a multiset over $S_N$, i.e.,
  $m\in\nat^{S_N}$.
  A Petri net is a pair $(N, m)$ where $N$ is a net and $m$ is a marking of
  $N$.  
\end{defi}

We denote $S_N\cup T_N$ by $N$, and omit the subscript $N$ if no
confusion arises.
We abbreviate a transition $\tr{t} \in T$ with {\it preset} $\preS{\tr
  t} = s_1$ and {\it postset} $\postS{\tr t} = s_2$ as $s_1 \pntrans
s_2$.
The preset and postset of a place  $\p{a} \in S$ are  
defined respectively as 
$\preS{\p{a}} = \{\tr{t}\ |\ \p{a} \in \postS{\tr{t}}\}$ 
and
$\postS{\p{a}{}} = \{\tr{t}\ |\ \p{a} \in \preS{\tr{t}}\}$.
We let $\minp N = \{x\in
N\ |\ \preS x = \emptyset\}$ and $\maxp N = \{x\in N |\ \postS x =
\emptyset\}$ denote the sets of \emph{initial} and \emph{final
elements} of $N$ respectively. 
Note that we only consider nets whose 
initial and final elements are places since 
transitions have non-empty presets and postsets, namely 
 $\preS{\tr{t}} \neq\emptyset$ and $\postS{\tr{t}} \neq\emptyset$ 
hold for all $\tr t$.

\begin{defi}[Net morphisms]
  Let $N, N'$ be nets. A pair $f =(f_S: S_N \rightarrow S_{N'}, f_T:
  T_N \rightarrow T_{N'})$ is a \emph{net morphism} from $N$ to $N'$
  (written $f: N \rightarrow N'$) if $f_S (\preS{\tr{t}}_{N}) =
  \preS{(f_T(\tr{t}))}_{N'}$ and $f_S (\postS{\tr{t}}_{N}) =
  \postS{(f_T(\tr{t}))}_{N'}$ for any $\tr{t}$. Moreover, we say $N$
  and $N'$ are \emph{isomorphic} if $f$ is bijective.
\end{defi}

The operational (interleaving) semantics of a Petri net is given by
the least relation on Petri nets satisfying the following inference
rule:
\[
  \mathrule{firing}
           {\tr t = m \;\pntrans\;  m' \in T_N}
           {(N, m\oplus m'')
            \red
            (N, m'\oplus m'')}
\]
which describes the evolution of the state of a net (represented by
the marking $m\oplus m''$) by the firing of a transition $m\pntrans
m'$ that consumes the tokens $m$ in its preset and produces the tokens
$m'$ in its postset. We shall call expressions of the form $(N,
m\oplus m'') \red (N, m'\oplus m'')$ as \emph{firings}.
We sometimes omit $\tr t$ in $\red$ when the fired transition is
uninteresting.

According to Definition~\ref{def:net}, each transition consumes and
produces at most one token in each place. On the other hand,
\textsc{p/t} nets below fetch and consume multiple tokens by defining
the pre- and postsets of transitions as multisets.
 
\begin{defi}[\textsc{p/t} net]\label{def:ptnet}
  A \emph{Place/Transition Petri net} (\textsc{p/t} net) is a 4-tuple
  $N = (S_N, T_N, \preS{\_}_N, \postS{\_}_N)$ where $S_N\subseteq
  \mathcal{P}$ is the (nonempty) set of places, $T_N \subseteq
  \transitions$ is the set of transitions and the functions
  $\preS{\_}_N, \postS{\_}_N: T_N\rightarrow \nat^{S_N}$ assign source
  and target to each transition such that $\preS{\tr{t}}
    \neq\emptyset$ and $\postS{\tr{t}} \neq\emptyset$ for all
    $\tr{t}\in T_N$.
  A marking of a net $N$ is multiset over $S_N$, i.e.,
  $m\in\nat^{S_N}$.
  A marked \textsc{p/t} net is a pair $(N, m)$ where $N$ is a
  \textsc{p/t} net and $m$ is a marking of $N$.
\end{defi}

The notions of pre- and postset, initial and final elements, morphisms
and operational semantics are straightforwardly extended to
\textsc{p/t} nets.
For technical reasons, we only consider \textsc{p/t} nets whose
transitions have non-empty pre- and post-sets.
Note that Petri nets can be regarded as a \textsc{p/t} net whose arcs
have unary weights.

Next, we introduce notation for sequences of firings (or
transitions). Let `;' denote concatenation of such sequences.
For the sequence $\fseq = \tr t_1 ;\tr t_2;\ldots;\tr t_n$, we write
$(N, m_0) \red[\fseq] (N, m_n)$ if $(N,m_0)\red[\tr t_1]
(N,m_1)\red[\tr t_2]\ldots\red[\tr t_n] (N,m_n)$; we call $\fseq$ a
firing sequence.
We write $(N, m_0) \red[]^{*} (N, m_n)$ if there exists $s$ such that
$(N, m_0) \red[\fseq] (N, m_n)$, and $\voidt_m$ for the empty sequence
with the marking $m$.

\begin{defi} Let $(N,m)$ be a \textsc{p/t}  net. 
The set of {\em reachable markings} $\reach{N, m}$ is defined as
$\big\{ m' \,|\, (N,m) \red[]^{*}\big (N, m')\}$.
\end{defi} 
We say a marked \textsc{p/t} net $(N, m)$ is {\em ($1$-)safe} if every
reachable marking is a set, i.e., $m'\in \reach{N,m}$ implies $m'\in
2^{S_N}$.

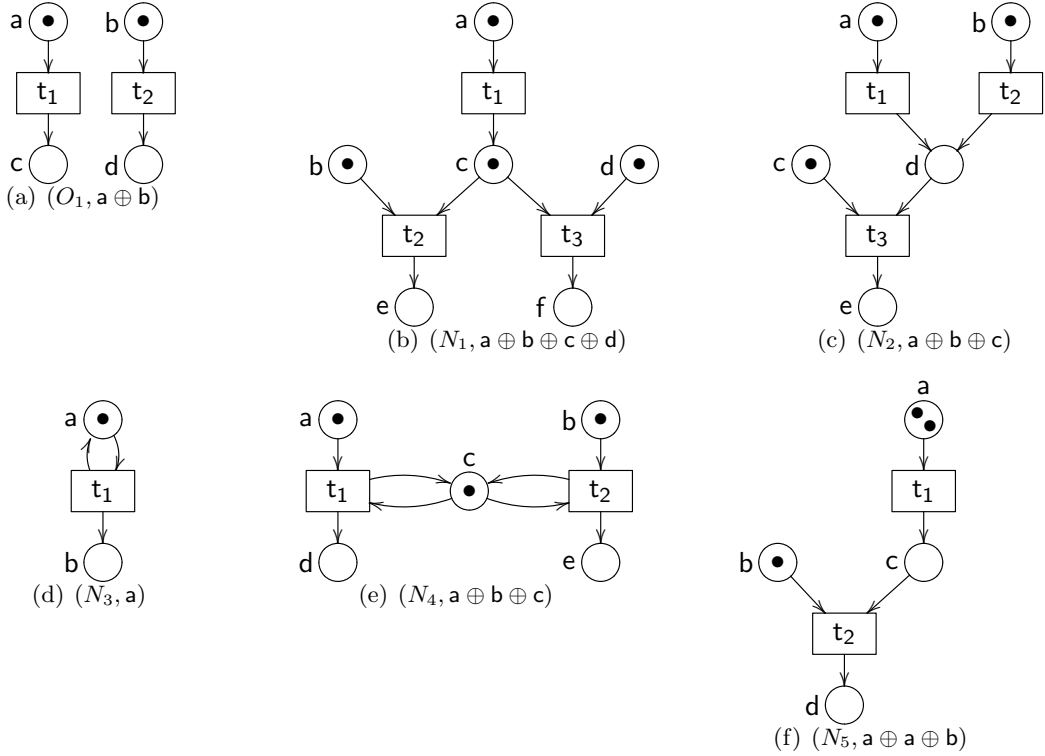
\begin{figure}[t]
\subfigure[$(O_1,\p a\oplus \p b)$]{ 
\label{fig:O1}
$$
 \xymatrix@R=1pc@C=1pc{
 \drawmarkedplace\ar[d]
 \nameplaceleft {\p a}
 &
 \drawmarkedplace\ar[d]
 \nameplaceleft {\p b}
 \\
 \drawtrans {\tr{t_1}} \ar[d] 
 &
 \drawtrans {\tr{t_2}} \ar[d] 
 \\
 \drawplace
 \nameplaceleft {\p c}
 &
 \drawplace
 \nameplaceleft {\p d}
 }
$$
}
\hspace{1cm}
\subfigure[$(N_1, \p a \oplus \p b \oplus \p c \oplus \p d)$]{ 
\label{fig:N1}
$$
 \xymatrix@R=1pc@C=.5pc{
 &
 &
 \drawmarkedplace\ar[d]
 \nameplaceleft {\p a}
 \\
 &
 &
 \drawtrans {\tr{t_1}} \ar[d] 
 \\
 \drawmarkedplace\ar[dr]
 \nameplaceleft {\p b}
 &
 &
 \drawmarkedplace\ar[dl]\ar[dr]
 \nameplaceleft {\p c}
 &
 &
 \drawmarkedplace\ar[dl]
 \nameplaceleft {\p d}
 &
 &
 \\
 &
  \drawtrans {\tr{t_2}} \ar[d] 
 &
 &
  \drawtrans {\tr{t_3}} \ar[d] 
 \\
 &
 \drawplace
 \nameplaceleft {\p e}
 &
 &
  \drawplace
 \nameplaceleft {\p f}
 }
$$
}
\subfigure[$(N_2, \p a \oplus \p b \oplus \p c)$]{ 
\label{fig:N2}
$$
 \xymatrix@R=1pc@C=.5pc{
 &
 \drawmarkedplace\ar[d]
 \nameplaceleft {\p a}
 &
 &
 \drawmarkedplace\ar[d]
 \nameplaceleft {\p b}
 \\
 &
 \drawtrans {\tr{t_1}} \ar[dr] 
 &
 &
 \drawtrans {\tr{t_2}} \ar[dl] 
 \\
\drawmarkedplace\ar[dr]
 \nameplaceleft {\p c}
 &
 &
 \drawplace\ar[dl]
 \nameplaceleft {\p d}
 &
 &
 \\
 &
  \drawtrans {\tr{t_3}} \ar[d] 
 &
 &
 \\
 &
 \drawplace
 \nameplaceleft {\p e}
 &
 }
$$
}
\qquad
\subfigure[$(N_3, \p a)$]{ 
\label{fig:N3}
$$
  \xymatrix@R=1pc@C=1pc{
 \drawmarkedplace\ar@/^/[d]
 \nameplaceleft {\p a}
  \\
 \drawtrans {\tr{t_1}}  \ar[d]\ar@/^/[u]
  \\
 \drawplace
 \nameplaceleft {\p b}
 }
$$
}
\qquad\qquad
\subfigure[$(N_4, \p a \oplus \p b \oplus \p c)$]{ 
\label{fig:N4}

$$
  \xymatrix@R=1pc@C=1pc{
 \drawmarkedplace\ar[d]
 \nameplaceleft {\p a}
 &&
 &&
 \drawmarkedplace\ar[d]
 \nameplaceleft {\p b}
 \\
 \drawtrans {\tr{t_1}} \ar[d] \ar@/^/[rr]
 &&
 \drawmarkedplace\ar@/_/[rr]\ar@/^/[ll]
 \nameplaceup {\p c} 
 &&
 \drawtrans {\tr{t_2}} \ar[d] \ar@/_/[ll]
 \\
 \drawplace
 \nameplaceleft {\p d}
 &&
 &&
 \drawplace
 \nameplaceleft {\p e}
 }
$$
}
\quad \quad
\subfigure[$(N_5, \p a \oplus \p a \oplus \p b)$]{ 
\label{fig:N5}
$$
 \xymatrix@R=1pc@C=.5pc{
 &&
  \drawplace\ar[d]
 \POS[]+<-0.2pc,.2pc>\drop{\bullet}
 \POS[]+<.2pc,-.2pc>\drop{\bullet}
 \nameplaceup{\p a}
 &
 &\\
 &&
 \drawtrans {\tr{t_1}} \ar[d] 
 &
 &
 \\
\drawmarkedplace\ar[dr]
 \nameplaceleft {\p b}
 &
 &
 \drawplace\ar[dl]
 \nameplaceleft {\p c}
 &
 &
 \\
 &
  \drawtrans {\tr{t_2}} \ar[d] 
 &
 &
 \\
 &
 \drawplace
 \nameplaceleft {\p d}
 &
 }
$$
}
\caption{\textsc{p/t} nets}\label{fig:nets}
\end{figure}

\begin{exa}
  \label{ex:nets}
  Figure~\ref{fig:nets} shows different \textsc{p/t} nets, which will
  be used throughout the paper.  As usual, places and transitions are
  represented by circles and boxes, respectively.
  The nets $O_1$ and $N_4$ are Petri nets, and $N_1$, $N_2$, $N_3$
  and $N_5$ are \textsc{p/t} nets which, when executing, may produce
  multiple tokens in some places.
\end{exa}

\subsection{Unfolding  of $\mbox{\sc p/t}$ nets}\label{sec:unf}

Our approach to reversing Petri nets relies on their occurrence net
semantics, which explicitly exhibits the causal ordering, concurrency,
and conflict among places and transitions.
We start by introducing several useful notions and notations.  First,
we shall describe a flow of causal dependencies in a net with the
relation $\prec$:
\begin{defi}\label{def:causaldep}
  Let $\prec$ be 
  $\{ (\p{a},\tr{t}) \mid {
    \p{a}\in S_N \wedge \tr t\in\postS{\p{a}}\}
    \ \cup\ 
    \{ (\tr{t},\p{a}) \mid
    \p{a}\in S_N \wedge \tr t\in \preS{\p{a}}\}}$.
  We write $\preceq$ for the reflexive and transitive closure of
  $\prec$.
\end{defi}
Consider Figure~\ref{fig:nets}.  We have $\p a \prec \tr t_1$ and $\tr
t_1 \prec \p c$ in $O_1$ as well as $\tr t_1 \preceq \tr t_2$ in
$N_1$.

Two transitions $\tr {t}_1$ and $\tr{t}_2$ are in an \emph{immediate
  conflict}, written $\tr t_1 \#_0 \tr{t}_2$, when $\tr{t}_1 \neq
\tr{t}_2$ and $\preS {\tr t_1} \cap \preS{\tr t_2} \neq \emptyset$.
For example, $\tr{t}_2$ and $\tr{t}_3$ in $N_1$ in
Figure~\ref{fig:nets} are in an immediate conflict since they share a
token in the place $\p c$. Correspondingly, for $\tr{t}_1$ and
$\tr{t}_2$ in $N_4$.
The \emph{conflict} relation $\#$ is defined by letting $x\#y$ if
there are $\tr t_1, \tr t_2 \in T$ such that $\tr t_1\preceq x$, and
$\tr t_2 \preceq y$, and $\tr{t}_1 \#_0 \tr{t}_2$.

We are now ready to define occurrence nets
following~\cite{NielsenPW81,HaymanW08}.

\begin{defi}[Occurrence net] A  net $(N,m)$ is an {\em occurrence net} if 
  \begin{enumerate}
  \item
    $N$ is \emph{acyclic};
  \item
    $N$ is a \emph{(1-)safe} net, i.e, any reachable marking is a set;
  \item
    $m = \minp N$, i.e., the initial marking is identified with the
    set of initial places;
  \item there are no \emph{backward conflicts}, i.e.,
    $|\preS{\p{a}}|\leq 1$ for all $a$ in $S_N$;
  \item there are no \emph{self-conflicts}, i.e, $\neg (\tr t \#\tr
    t)$ for all $\tr t$ in $T_N$.
  \end{enumerate}
\end{defi}
We use $O, O', \ldots$ to range over occurrence nets.

\begin{exa}
  The net $O_1$ in Figure \ref{fig:nets} is an occurrence net, while
  the remaining nets are not.
  $N_1$ is not an occurrence net since there is a token in place $\p
  c$ and $\p c$ is not an initial place of the net.
  $N_2$ has a backward conflict since two transitions produce tokens
  on the place $\p d$.
  $N_3$ is cyclic, and $N_4$ is cyclic and has a backward conflict on
  $\p c$.
  $N_5$ is not an occurence net since it is not 1-safe.
\end{exa}

The absence of backward conflicts in occurrence nets ensures that each
place appears in the postset of at most one transition.
Hence, pre- and postset relations can be interpreted as a causal
dependency. So, $\preceq$ represents causality.

We say $x, y \in S_N \cup T_N$ are \emph{concurrent}, written
$x\ co\ y$, if $x\neq y$ and $x\not\preceq y$, $y\not\preceq x$, and
$\neg x \# y$.
A set $X\subseteq S_N \cup T_N$ is concurrent, written $CO(X)$, if
$\forall x, y\in X: x\neq y \Rightarrow x\ co\ y$, and $|\{\tr{t}\in
T_N\ |\ \exists x \in X, \tr{t} \preceq x\}|$ is finite.
For example, the set $\{\tr{t_1},\tr{t_2}\}$ of firings in $O_1$ of
Figure~\ref{fig:nets} is concurrent, so we can write $CO(\{\tr{t_1},
\tr{t_2}\})$.

The notions of causality, conflict and concurrency for firings are
defined in terms of the corresponding notions for their
transitions.
For example, the firing $(N,m)\red (N,m')$ causes the firing
$(N,m_1)\red[\tr t'] (N,m_1')$, written as $(N,m)\red (N,m')\preceq
(N,m_1)\red[\tr t'] (N,m_1')$, if $\tr t\preceq \tr
t'$.
Correspondingly for the conflict and concurrency relations.

Two firings (or transitions) are \textit{coinitial} if they start with
the same marking, and \textit{cofinal} if they end up in the same
marking.
We now have a simple version of Square Lemma~\cite{rccs} for forward
concurrent firings.  It will be helpful in proving our
Lemma~\ref{lem:SL-all} in Section~4.

\begin{lem}\label{con=noconflict}
Let $O$ be an occurrence net and $\tr t$ and $\tr t'$ be
enabled coinitial transitions of $O$.
Then $\tr t\ co\ \tr t'$ if and only if $\preS{\tr t}\cap \preS{\tr
  t'}=\emptyset$.
\end{lem}

\begin{proof} 
$\Rightarrow$ part: $\tr t\ co\ \tr t'$ implies $\neg \tr t\#_0\tr
  t'$. Hence, $\preS{\tr t}\cap \preS{\tr t'}=\emptyset$.
 
\noindent 
$\Leftarrow$ part: We proceed by contradiction. Assume that $\tr t$
and $\tr t'$ are enabled at some marking $m$ of $O$ and $\neg(\tr
t\ co\ \tr t')$.
Since $\preS{\tr t}\cap \preS{\tr t'}=\emptyset$, either $\tr t\preceq
\tr t'$ or $\tr t'\preceq \tr t$ holds (note that $\tr t\preceq \tr
t'$ and $\tr t'\preceq \tr t$ do not hold simultaneously because
occurrence nets are acyclic).
We consider the case $\tr t\preceq \tr t'$ and show that $t$ is
enabled implies that $t'$ is not enabled.
Since, $\tr t\preceq \tr t'$, there exists $\p a\in\preS{\tr t'}$ such
that $\tr t \prec \p a$.
Therefore, $\p a\not\in\minp O$. Finally, $|\preS{\p{a}}|\leq 1$ and
$\tr t \prec \p a$ imply $\p{a} \not\in m$, which is in contradiction
with the assumption that $\tr t'$ is enabled at $m$.
The case $\tr t'\preceq \tr t$ follows analogously.
\end{proof}

\begin{lem}\label{lem:SL-F}
Let $O$ be an occurrence net and $(O,n) \red (O,n')$ and
$(O,n)\red[\tr t'] (O,n'')$ be coinitial concurrent firings.
Then, there exist firings $(O,n')\red[\tr t'] (O,n'_1)$ and
$(O,n'')\red (O,n_1'')$, for some $n'_1$ and $n_1''$, such that the
firings are cofinal.
\end{lem}

\begin{proof}
Assume that $(O,n) \red (O,n')$ and $(O,n)\red[\tr t'] (O,n'')$ are
coinitial concurrent firings.
That means $\tr t$ and $\tr t'$ are coinitial, concurrent and are
enabled. Hence they do not cause each other, and they are also not in
an immediate conflict.
So, we can write the firings as
$(O, m\oplus m'\oplus m'')\red (O,m_1\oplus m'\oplus m'')$ and
$(O, m\oplus m'\oplus m'')\red[\tr t'](O, m\oplus m_1'\oplus m'')$,
where $\preS{\tr t}\subseteq m$ and $\preS{\tr t'}\subseteq m'$ with
markings $m$ and $m'$ not overlapping.
Clearly then these are valid firings:
$(O, m_1\oplus m'\oplus m'')\red[\tr t'] (O, m_1\oplus m_1'\oplus m'')$
and
$(O, m\oplus m_1'\oplus m'')\red[\tr t] (O, m_1\oplus m_1'\oplus m'')$.
They are cofinal since they end up in $(O, m_1\oplus m_1'\oplus m'')$.
\end{proof}

This lemma can be equivalently expressed as follows: if $\tr t$ and
$\tr t'$ are coinitial and concurrent, then $\tr t;\tr t'$ and $\tr
t';\tr t$ are cofinal.
Informally, if firings with transitions $\tr t$ and $\tr t'$ originate
from one corner of a square, and if they represent independent
(concurrent) events, then the square completes with two firings $\tr
t'$ and $\tr t$, which meet at the opposite corner of the square.
Hence, the order in which concurrent transitions are executed in a
firing sequence does not matter.
We then consider sequences equivalent up to the swapping of concurrent
transitions. This corresponds to considering the set of \m traces
induced by $co$ as the independence relation.

Formally, trace equivalence $\equiv$ is the least congruence over
firing sequences $s$ such that $\forall \tr t_1, \tr{t}_2 : \tr
t_1\ co\ \tr t_2 \implies \tr t_1 ; \tr t_2 \equiv \tr t_2;\tr t_1$.
The equivalence classes of $\equiv$ are the (Mazurkiewicz)
\emph{traces}. We use $\trace$ to range over such traces. We also will
use $\voidt$ for the empty trace, and $;$ for the concatenation
operator.

For occurrence nets we have this standard property:
\[
\fseq_1 \equiv \fseq_2\
\textit{ iff }\
(O,m_0)\red[\fseq_1](O,m_n) \iff (O,m_0)\red[\fseq_2](O,m_n)
\]
Two traces are \textit{coinitial} if they start with the same marking,
and \textit{cofinal} if they end up in the same marking.
Hence, the above property tells us that two traces that are
\textit{coinitial} and \textit{cofinal} are precisely trace
equivalent.

The unfolding of a net $N$ is the least occurrence net that can
account for all the possible computations of $N$ and makes explicit
causal dependencies, conflicts and concurrency between
firings~\cite{NielsenPW81}.
\begin{figure}[t]
  \[
    \begin{array}{l}
      \mathrule{
        ini-mk}{m(\p{a}) = n
      }{
        \{\p{a}(\emptyset, i) \ |\ 1\leq i \leq n\} \subseteq S
      }
      \\[20pt]
      \mathrule{pre}{
        H = \{\p a_j (h_j, i_j)\ |\ j \in J\} \subseteq S
        \qquad
        CO(H)
        \qquad
        \tr{t} \in T_N
        \quad
        \preS{\tr{t}}_N = \oplus_{j\in J} \p a_j
      }{
        \tr{t}(H) \in T ,\quad \preS(\tr{t}(H)) =  H
      }
      \\[10pt]
      \mathrule{post}{
        \quad x = \tr{t}(H) \in T
      }{
        Q=\{\p a (\{x\}, i) \mid 1 \leq i \leq \postS{\tr{t}}_N(\p a)\}
        \subseteq S,\quad  \postS x = Q
      }
    \end{array}
    \]
    \caption{Unfolding rules.}\label{fig:unfolding}
\end{figure}

\begin{defi} [Unfolding]
Let $(N,m)$ be a \textsc{p/t} net. The unfolding of $N$ is the
occurrence net $\mathcal{U}\lbrack N,m\rbrack = (S,T,\preS{\_},
\postS{\_})$ generated inductively by the inference rules in
Figure~\ref{fig:unfolding}, and the folding morphism $(f_S, f_T) :
\mathcal{U}\lbrack N,m\rbrack \rightarrow N$ is defined such that
$f_S(\p a( \_, \_)) = \p a$ and $f_T(\tr t( \_)) = \tr t$.
\end{defi}

Places in the unfolding of a net represent tokens and are named by
triples $\p{a}(H, i)$ where: $\p{a}$ is the place of the original net
$N$ in which the token resides; $H$ is the set of its immediate causes
(i.e., the causal history of the token); and $i$ is a positive integer
used to disambiguate tokens with the same history, i.e., when the
initial marking assigns several tokens to a place or a transition
produces multiple tokens in a place.
Analogously, transitions in the unfolding represent firings or events
and are named by pairs $\tr{t}(H)$, where $H$ encodes the causal
history as above and $\tr{t}$ is the fired transition.

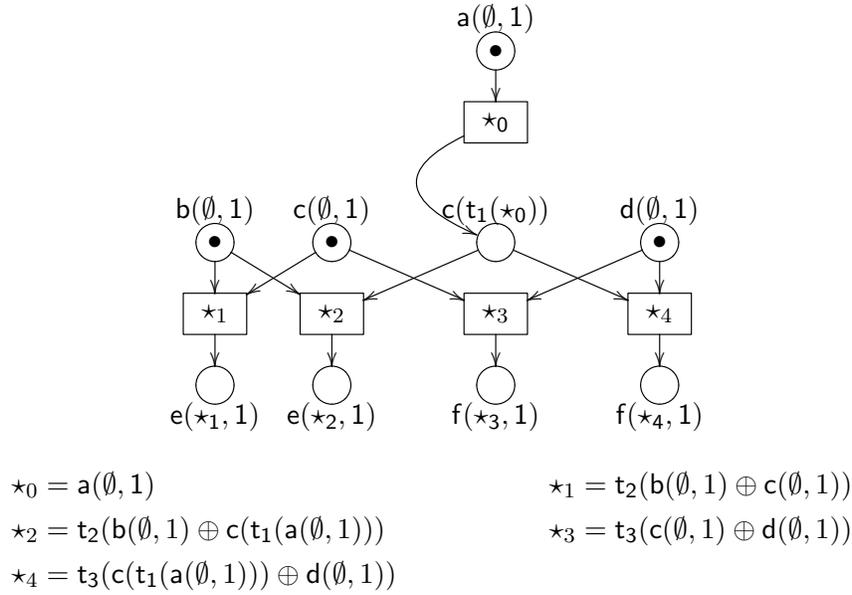
\begin{figure}[t]
  \subfigure{
    \xymatrix@R=1pc@C=.2pc{
      &&
      &
      &&
      &&&
      \drawmarkedplace\ar[d]\nameplaceup{\p{a(\emptyset,1)}}
      \\
      &&
      &
      &&
      &&&
      \drawtrans{\p{\star_0}}\ar@/_2.5pc/[dd]
      \\\\
      \drawmarkedplace\ar[d]\ar[drrr]
      \nameplaceup{\p{b(\emptyset,1)}}
      &&
      &
      \drawmarkedplace\ar[dlll] \ar[drrrrr]
      \nameplaceup{\p{c(\emptyset,1)}}
      &&
      &&&
      \drawplace\ar[dlllll]\ar[drrrrr]
      \nameplaceup{\p{c(\tr{t_1}(\star_0))}}
      &&
      &&&
      \drawmarkedplace\ar[dlllll]\ar[d]
      \nameplaceup{\p{d(\emptyset,1)}} 
      \\
      \drawtrans{\star_1} \ar[d]
      &&
      &
      \drawtrans{\star_2} \ar[d]
      &&
      &&&
      \drawtrans{\star_3} \ar[d]
      &&
      &&&
      \drawtrans{\star_4}\ar[d]
      \\
      \drawplace\nameplacedown{\p{e(\star_1,1)}}
      &&
      &
      \drawplace\nameplacedown{\p{e(\star_2,1)}}
      &&
      &&&
      \drawplace\nameplacedown{\p{f(\star_3,1)}}
      &&
      &&&
      \drawplace\nameplacedown{\p{f(\star_4,1)}} 
    }
  }
  
  \begin{align*}
    &
    \star_0 = \p{a(\emptyset,1)
    }&
    \star_1 = \tr{t_2}(\p{b}(\emptyset,1)\oplus \p{c}(\emptyset,1))
    &
    \\
    &
    \star_2 = \tr{t_2}({
      \p{b}(\emptyset,1)\oplus \p{c}(\tr{t_1}(\p{a}(\emptyset,1))
    }
    )
    &
    \star_3 = \tr{t_3}(\p{c}(\emptyset,1)\oplus \p{d}(\emptyset,1))
    &
    \\
    &
    \star_4 = \tr{t_3}({
      \p{c}(\tr{t_1}(\p{a}(\emptyset,1))) \oplus \p{d}(\emptyset,1)
    }
    )
  \end{align*}
\caption{Unfolding of $N_1$ with histories}
\label{fig:unfold_hist}
\end{figure}

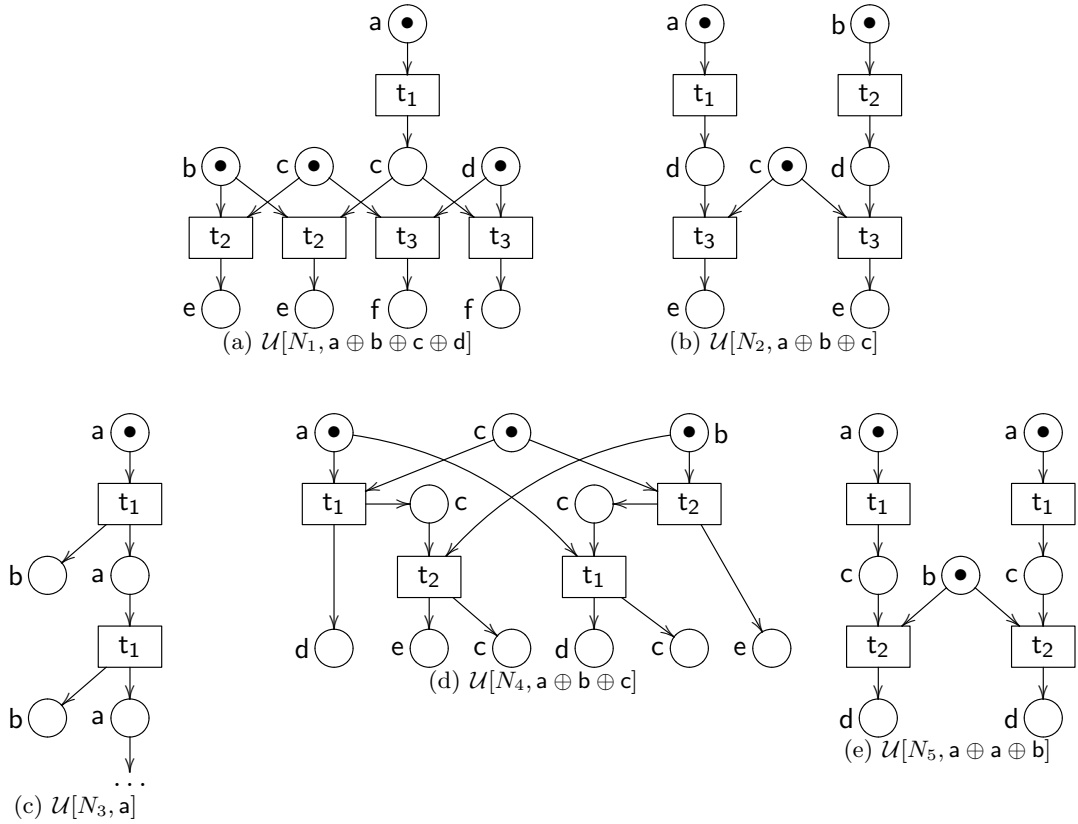
\begin{figure}[t]
  \subfigure[
    $\mathcal{U}\lbrack N_1, \p a\oplus\p b\oplus\p c\oplus\p d\rbrack$]{ 
    \label{fig:UN1}
    $$
    \xymatrix@R=1pc@C=.2pc{
      &&&&
      \drawmarkedplace \ar[d]\nameplaceleft{\p a}
      \\
      &&&&
      \drawtrans{\tr{t_1}}\ar[d]
      \\
      \drawmarkedplace \ar[d] \ar[drr]\nameplaceleft{\p b}
      &&
      \drawmarkedplace \ar[dll] \ar[drr]\nameplaceleft{\p c}
      &&
      \drawplace \ar[dll]\ar[drr]\nameplaceleft{\p c}
      &&
      \drawmarkedplace \ar[d] \ar[dll] \nameplaceleft{\p d}
      \\
      \drawtrans{\tr{t_2}}\ar[d]
      &&
      \drawtrans{\tr{t_2}}\ar[d]
      &&
      \drawtrans{\tr{t_3}}\ar[d]
      &&
      \drawtrans{\tr{t_3}}\ar[d]
      \\
      \drawplace \nameplaceleft{\p e}
      &&
      \drawplace \nameplaceleft{\p e}
      &&
      \drawplace \nameplaceleft{\p f}
      &&
      \drawplace \nameplaceleft{\p f} 
    }
    $$
  }
  \hspace{1cm}
  \subfigure[$\mathcal{U}\lbrack N_2, \p a \oplus \p b \oplus \p c\rbrack$]{ 
    \label{fig:UN2}
    $$
    \xymatrix@R=1pc@C=1pc{
      \drawmarkedplace\ar[d]
      \nameplaceleft {\p a}
      &
      &
      \drawmarkedplace\ar[d]
      \nameplaceleft {\p b}
      \\
      \drawtrans {\tr{t_1}} \ar[d] 
      &
      &
      \drawtrans {\tr{t_2}} \ar[d] 
      \\
      \drawplace\ar[d]
      \nameplaceleft {\p d}
      &
      \drawmarkedplace\ar[dl]\ar[dr]
      \nameplaceleft {\p c}
      &
      \drawplace\ar[d]
      \nameplaceleft {\p d}
      \\
      \drawtrans {\tr{t_3}} \ar[d] 
      &
      &
      \drawtrans {\tr{t_3}} \ar[d] 
      \\
      \drawplace
      \nameplaceleft {\p e}
      &&
      \drawplace
      \nameplaceleft {\p e}
    }
    $$
  }
  \\[.4cm]
  \subfigure[$\mathcal{U}\lbrack N_3, \p a\rbrack$]{ 
    \label{fig:UN3}
    $$
    \xymatrix@R=1pc@C=1pc{
      &
      \drawmarkedplace\ar[d]
      \nameplaceleft {\p a}
      \\
      &
      \drawtrans {\tr{t_1}}  \ar[d]\ar[dl]
      \\
      \drawplace
      \nameplaceleft {\p b}
      &
      \drawplace\ar[d]
      \nameplaceleft {\p a}
      \\
      &
      \drawtrans {\tr{t_1}}  \ar[d]\ar[dl]
      \\
      \drawplace
      \nameplaceleft {\p b}
      &
      \drawplace\ar[d]
      \nameplaceleft {\p a}
      \\
      &{\ldots} }
    $$
  }
  \hspace{1cm}
  \subfigure[$\mathcal{U}\lbrack N_4, \p a \oplus \p b \oplus \p c \rbrack$]{ 
    \label{fig:UN4}
    $$
    \xymatrix@R=1pc@C=1pc{
      \drawmarkedplace \ar[d]\ar@/^1pc/[ddrrr]
      \nameplaceleft {\p a}
      &&
      \drawmarkedplace \ar[dll] \ar[drr]
      \nameplaceleft {\p c}
      &&
      \drawmarkedplace \ar[d] \ar@/_1pc/[ddlll]
      \nameplaceright {\p b}
      \\
      \drawtrans {\tr{t_1}}  \ar[dd] \ar[r]
      &
      \drawplace \ar[d]\nameplaceright {\p c}
      &&
      \drawplace \ar[d]\nameplaceleft {\p c}
      &
      \drawtrans {\tr{t_2}}  \ar[l]\ar[ddr]
      \\
      &
      \drawtrans {\tr{t_2}} \ar[d] \ar[dr]
      &&
      \drawtrans {\tr{t_1}} \ar[d] \ar[dr]
      \\
      \drawplace \nameplaceleft {\p d}
      &
      \drawplace \nameplaceleft {\p e}
      &
      \drawplace \nameplaceleft {\p c}
      &
      \drawplace \nameplaceleft {\p d}
      &
      \drawplace \nameplaceleft {\p c}
      &
      \drawplace \nameplaceleft {\p e}   
    }
    $$
  }
  \subfigure[$\mathcal{U}\lbrack N_5, \p a \oplus \p a \oplus \p b\rbrack$]{ 
    \label{fig:UN5}
    $$
    \xymatrix@R=1pc@C=1pc{
      \drawmarkedplace\ar[d]
      \nameplaceleft {\p a}
      &
      &
      \drawmarkedplace\ar[d]
      \nameplaceleft {\p a}
      \\
      \drawtrans {\tr{t_1}} \ar[d] 
      &
      &
      \drawtrans {\tr{t_1}} \ar[d] 
      \\
      \drawplace\ar[d]
      \nameplaceleft {\p c}
      &
      \drawmarkedplace\ar[dl]\ar[dr]
      \nameplaceleft {\p b}
      &
      \drawplace\ar[d]
      \nameplaceleft {\p c}
      \\
      \drawtrans {\tr{t_2}} \ar[d] 
      &
      &
      \drawtrans {\tr{t_2}} \ar[d] 
      \\
      \drawplace
      \nameplaceleft {\p d}
      &&
      \drawplace
      \nameplaceleft {\p d}
    }
    $$
  }
  \caption{Unfoldings of \textsc{p/t} nets}\label{fig:unfoldings}
\end{figure}

\begin{exa}
The unfolding of the net $(N_1, \p a \oplus \p b \oplus \p c \oplus \p
d)$ in Figure~\ref{fig:nets} is given in Figure~\ref{fig:unfold_hist}.
\end{exa}
Henceforth, we adopt the usual convention of omitting causal histories
in the name of places and transitions when depicting unfoldings, and
can write instead the image of the unfolding morphisms, as illustrated
by the following example.

\begin{exa}
The unfoldings of the nets $(N_1, \p a \oplus \p b \oplus \p c \oplus
\p d)$, $(N_2,\p a \oplus \p b \oplus \p c)$, $(N_3, \p a)$, $(N_4,\p
a \oplus \p b \oplus \p c)$ and $(N_4,\p a \oplus \p a \oplus \p b)$
in Figure~\ref{fig:nets} are shown in
Figure~\ref{fig:unfoldings}. Since $O_1$ in Figure~\ref{fig:nets} is
an occurrence net its unfolding is isomorphic to $O_1$, thus it is
omitted.
Consider the occurrence net $\mathcal{U}\lbrack N_1, \p a \oplus \p b
\oplus \p c \oplus \p d\rbrack$.  The leftmost transition $\tr t_2$ is
different from the other transition $\tr t_2$ since they have
different histories: the leftmost $\tr t_2$ is caused by the tokens in
$\p b$ and $\p c$ (which are available in the initial marking),
whereas the other $\tr t_2$ is caused only by the token in $\p b$ and
the token that is produced by the firing of $\tr t_1$.
Correspondingly, for the two transitions labelled $\tr t_3$.  Consider
$\mathcal{U}\lbrack N_2, \p a \oplus \p b \oplus \p c\rbrack$. After
the transitions $\tr t_1$ and $\tr t_2$ have fired, there is a token
in each of the places labelled $\p d$. The token in the leftmost $\p
d$ has the history corresponding to the firing of $\tr t_1$ and the
token in the other $\p d$ has the history corresponding to $\tr t_2$.
Once $\tr t_3$ has fired, we can tell the copies of $\tr t_3$ apart by
inspecting their histories: the leftmost $\tr t_3$ is caused by a
token in $\p d$ with the history $\tr t_1$ (as well as the token in
$\p c$), whereas the other $\tr t_3$ is caused by $\p d$ with the
history $\tr t_2$ and by $\p c$.
\end{exa}

\section{Reversing Occurrence Nets}

\begin{defi}\label{lbl:rev_occ}
Let $O$ be an occurrence net. The reversible version of $O$ is $\rev O
= (S_{\rev O}, T_{\rev O}, \preS{\_}_{\rev O}, \postS{\_}_{\rev O})$
defined as follows:
\[
\begin{array}{l@{\hspace{1cm}}l}
  S_{\rev O} = S_O
  &
  T_{\rev O} = T_O \cup \{\rev {\tr t} \ |\ \tr{t} \in T_O\}
  \\[10pt]
  \preS{\tr t}_{\rev O} = 
  \begin{cases}
    \preS{\tr t}_O & \mbox{if}\ \tr t\in T_O \\
    \postS{\tr t}_O & \mbox{otherwise}
  \end{cases}
  &
  \postS{\tr t}_{\rev O} = 
  \begin{cases}
    \postS{\tr t}_O & \mbox{if}\ \tr t\in T_O \\
    \preS{\tr t}_O & \mbox{otherwise}
  \end{cases}
\end{array}
\]
\end{defi}
Given a transition $\tr{t}$ we write $\rev {\tr t}$ for a transition
that reverses $\tr t$.
We shall call transitions like $\rev {\tr t_1}$ and $\rev {\tr t_2}$
in Figure~\ref{fig:rev-nets} \emph{reverse} (or backwards) transitions
and transitions like $\tr{t}$ \emph{forward} transitions.
We will use $t, t_1$, $t_2, \ldots$ to denote forward or reverse
transitions.  If $t$ is a reverse transition, say $\rev{\tr{t}}$, then
$\rev{t}$ is the forward transition $\tr{t}$.

\begin{exa}
The reversible version of the nets in Figure~\ref{fig:nets} are shown
in Figure~\ref{fig:rev-nets}.
We remark that they are the reversible versions of the nets in
Figure~\ref{fig:unfoldings}, which are the unfoldings of the original
nets.
\begin{figure}[ht]
  \subfigure[$(\rev{O_1},\p a\oplus \p b)$]{ 
    \label{fig:revO1}
    $$
    \xymatrix@R=1pc@C=1pc{
      \drawmarkedplace\ar[d]
      \nameplaceleft {\p a}
      &&
      \drawmarkedplace\ar[d]
      \nameplaceleft {\p b}
      &  
      \\
      \drawtrans {\tr{t_1}} \ar[d] 
      &
      \drawtrans {\rev{\tr{t_1}}} \ar@{-->}@/_2ex/[ul]
      &
      \drawtrans {\tr{t_2}} \ar[d] 
      &
      \drawtrans {\rev{\tr{t_2}}} \ar@{-->}@/_2ex/[ul]
      \\
      \drawplace\ar@{-->}@/_2ex/[ur] 
      \nameplaceleft {\p c}  
      &&
      \drawplace \ar@{-->}@/_2ex/[ur] 
      \nameplaceleft {\p d}
    }
    $$
  }
  \hspace{1.5cm}
  \subfigure[$(\rev{N_1},\p a \oplus \p b \oplus \p c \oplus \p d$)]{ 
    \label{fig:RevN1}
    $$
    \xymatrix@R=1pc@C=.8pc{
      &&&&
      \drawmarkedplace \ar@{<--}[dl] \ar[d]\nameplaceleft{\p a}
      \\
      \drawtrans{\rev{\tr{t_2}}}\ar@{-->}[d]\ar@{-->}[drr] 
      &&
      \drawtrans{\rev{\tr{t_2}}}\ar@{-->}[drr]\ar@{-->}[dll] 
      &
      \drawtrans{\rev{\tr{t_1}}}\ar@{<--}[dr]& \drawtrans{\tr{t_1}}\ar[d]
      &&
      \drawtrans{\rev{\tr{t_3}}} \ar@{-->}[d] \ar@{-->}[dll]
      \\
      \drawmarkedplace \ar[d] \ar[drr]\nameplaceleft{\p b}
      &&
      \drawmarkedplace \ar[dll] \ar[drr]\nameplaceleft{\p c}
      &&
      \drawplace \ar[dll]\ar[drr]\nameplaceleft{\p c}
      &&
      \drawmarkedplace \ar[d] \ar[dll] \nameplaceleft{\p d}
      \\
      \drawtrans{\tr{t_2}}\ar[d]
      &&
      \drawtrans{\tr{t_2}}\ar[d]
      &&
      \drawtrans{\tr{t_3}}\ar[d]
      &&
      \drawtrans{\tr{t_3}}\ar[d]
      \\
      \drawplace  \ar@{-->}@/^6ex/[uuu]\nameplaceleft{\p e}
      &&
      \drawplace  \ar@{-->}@/^6ex/[uuu]\drawplace \nameplaceleft{\p e}
      &&
      \drawplace \nameplaceleft{\p f}
      &&
      \drawplace  \ar@{-->}@/_6ex/[uuu] \nameplaceleft{\p f}
      \\
      &&&&
      \drawtrans{\rev{\tr{t_3}}}\ar@{<--}[u] \ar@{-->}[uuull]
      \ar@{<--}[u] \ar@{-->}[uuurr]
      &&
    }
    $$
  }
  \subfigure[($\rev{N_2}, \p a \oplus \p b \oplus \p c)$]{ 
    \label{fig:RevN2}
    $$
    \xymatrix@R=1pc@C=1pc{
      &
      \drawmarkedplace\ar[d]
      \nameplaceleft {\p a}
      &
      &
      \drawmarkedplace\ar[d]
      \nameplaceleft {\p b}
      \\
      \drawtrans {\rev{\tr{t_1}}}  \ar@{-->}[ur]  \ar@{<--}[dr]
      &
      \drawtrans {\tr{t_1}} \ar[d] 
      &
      &
      \drawtrans {\tr{t_2}} \ar[d]
      & 
      \drawtrans {\rev{\tr{t_2}}} \ar@{-->}[ul]  \ar@{<--}[dl]
      \\
      &
      \drawplace\ar[d]
      \nameplaceleft {\p d}
      &
      \drawmarkedplace\ar[dl]\ar[dr]
      \nameplaceleft {\p c}
      &
      \drawplace\ar[d]
      \nameplaceleft {\p d}
      \\
      &
      \drawtrans {\tr{t_3}} \ar[d] 
      &
      &
      \drawtrans {\tr{t_3}} \ar[d] 
      \\
      &
      \drawplace \ar@{-->}[d] \nameplaceleft  {\p e}
      &&
      \drawplace  \ar@{-->}[d]  \nameplaceleft {\p e}
      \\
      &
      \drawtrans {\rev{\tr{t_3}}} \ar@{-->}@/_2ex/[uuur]
      \ar@{-->}@/^5ex/[uuu]   
      &&
      \drawtrans {\rev{\tr{t_3}}}\ar@{-->}@/^2ex/[uuul] \ar@{-->}@/_5ex/[uuu]
    }
    $$
  }
  \hspace{2cm}
  \subfigure[($\rev{N_3}, \p a$)]{ 
    \label{fig:RevN3}
    $$
    \xymatrix@R=1pc@C=1pc{
      &&
      \drawmarkedplace\ar[d]
      \nameplaceleft {\p a}
      \\
      &
      \drawtrans {\rev{\tr{t_1}}} \ar@{-->}[ur]  
      &
      \drawtrans {\tr{t_1}}  \ar[d]\ar[dl]
      \\
      &
      \drawplace \ar@{-->}[u]
      \nameplaceleft {\p b}
      &
      \drawplace\ar[d]\ar@{-->}[ul]
      \nameplaceleft {\p a}
      \\
      &
      \drawtrans {\rev{\tr{t_1}}} \ar@{-->}[ur]  
      &
      \drawtrans {\tr{t_1}}  \ar[d]\ar[dl]
      \\
      &
      \drawplace \ar@{-->}[u]
      \nameplaceleft {\p b}
      &
      \drawplace\ar[d]\ar@{-->}[ul]
      \nameplaceleft {\p a}
      \\
      &
      \ldots \ar@{-->}[ur] 
      &
      \ldots
    }
    $$
  }
  \subfigure[$ (\rev{N_4}, \p a \oplus \p b \oplus \p c \rbrack$]{ 
    \label{fig:RevN4}
    $$
    \xymatrix@R=1.3pc@C=1.5pc{
      \drawtrans{\rev{\tr{t_1}}} \ar@{-->}[d]  \ar@{-->}@/^1pc/[drr]
      &&
      &&  
      \drawtrans{\rev{\tr{t_2}}} \ar@{-->}[d] \ar@{-->}@/_1pc/[dll]
      \\
      \drawmarkedplace \ar[d] \ar@/^1pc/[ddrrr]
      \nameplaceleft {\p a}
      &&
      \drawmarkedplace \ar[dll] \ar[drr]
      \nameplaceleft {\p c}
      &&
      \drawmarkedplace \ar[d] \ar@/_1pc/[ddlll]
      \nameplaceright {\p b}
      \\
      \drawtrans {\tr{t_1}}  \ar[dd] \ar[r]
      &
      \drawplace \ar[d] \ar@{-->}[uul] \nameplaceright {\p c}
      &&
      \drawplace \ar[d] \ar[d] \ar@{-->}[uur] \nameplaceleft {\p c}
      &
      \drawtrans {\tr{t_2}}  \ar[l]\ar[ddr]
      \\
      &
      \drawtrans {\tr{t_2}} \ar[d] \ar[dr]
      &&
      \drawtrans {\tr{t_1}} \ar[d] \ar[dr]
      \\
      \drawplace  \ar@{-->}@/^5ex/[uuuu]
      \nameplaceleft {\p d}
      &
      \drawplace \ar@{-->}[d] \nameplaceleft {\p e}
      &
      \drawplace \ar@{-->}[dl] \nameplaceleft {\p c}
      &
      \drawplace \ar@{-->}[d]  \nameplaceleft {\p d}
      &
      \drawplace \ar@{-->}[dl] \nameplaceleft {\p c}
      &
      \drawplace \ar@{-->}@/_5ex/[uuuul] \nameplaceleft {\p e}
      \\
      &
      \drawtrans {\rev{\tr{t_2}}} \ar@{-->}@/^5ex/[uuu]
      \ar@{-->}@/^5ex/[uuuurrr] 
      &&
      \drawtrans {\rev{\tr{t_1}}} \ar@{-->}@/_4ex/[uuu]
      \ar@{-->}@/_5ex/[uuuulll]
    }
    $$
  }
  \caption{Reversible \textsc{p/t} and Petri nets}\label{fig:rev-nets}
\end{figure}
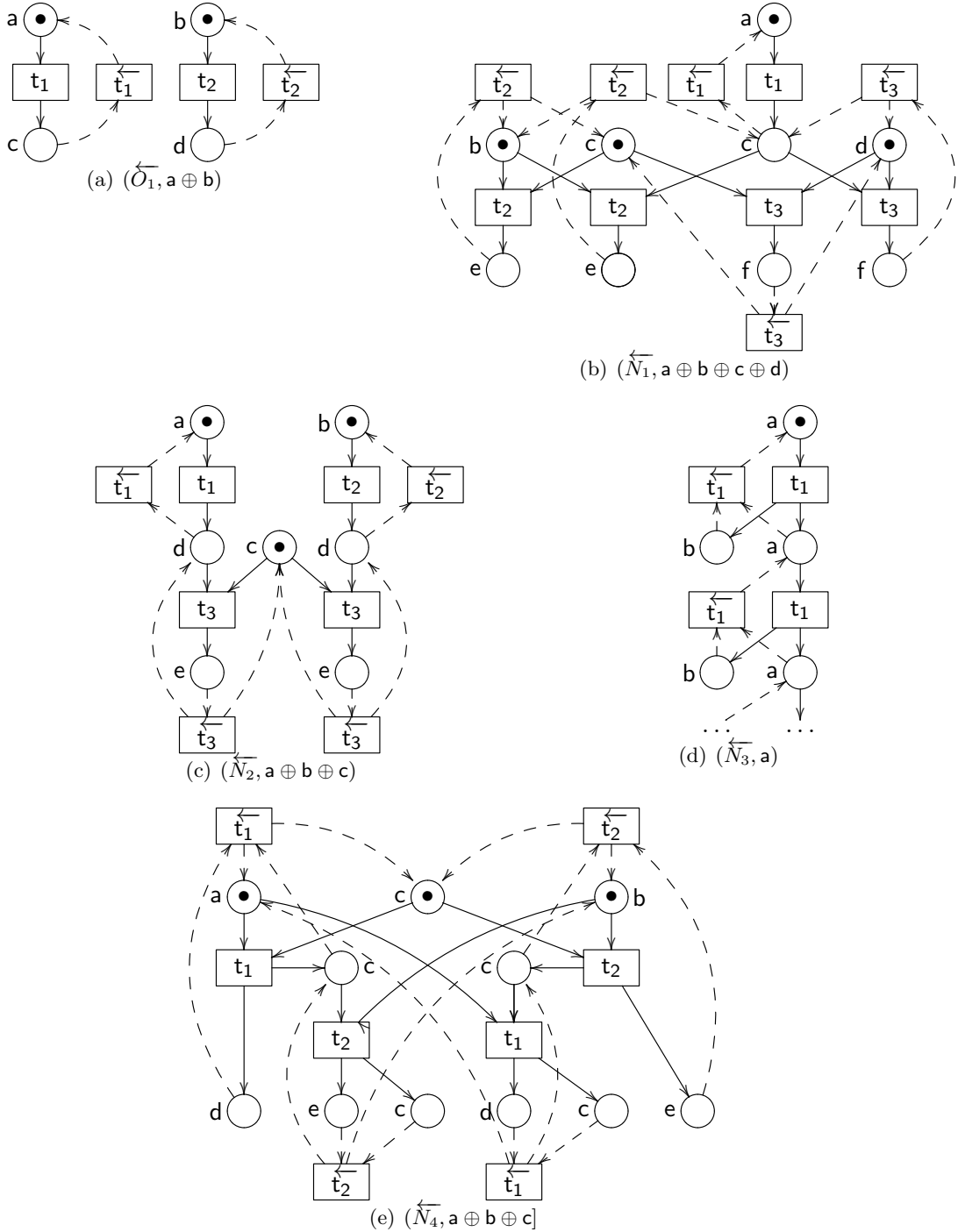

\end{exa}

Given $\rev O$, we write $(\rev O, m) \fred[t] (\rev O,m')$ for a
forward firing when $t\in T_O$, and $(\rev O, m) \bred[t] (\rev O,m')$
for the reverse (or backward) firing when $t\not\in T_O$.
We let $\red[t]$ be $\fred[t]\cup\bred[t]$.
Henceforth we often refer to firings $(\rev O, m) \red[t] (\rev
O,m')$, $(\rev O, m) \fred (\rev O,m')$ and $(\rev O, m)
\bred[\rev{\tr{t}}] (\rev O,m')$ simply as $t$, $\tr{t}$ and
$\rev{\tr{t}}$ respectively, especially when discussing properties of
$\red[t]$ in this section and in Section~\ref{sec:prop}.
We shall work with sequences of firings, ranged over by $s, s_1$ and
$s_2$.
We say that a sequence is a {\em forward} (respectively {\em
  backward}) {\em sequence} when all its firings are forward
(respectively backward).

Next, we extend the notions of causality, conflict and concurrency to
forward transitions and reverse transitions in reverse
versions of occurrence nets.
We extend $\prec$ in Definition~\ref{def:causaldep} to cover reverse 
transitions in an obvious way using Definition~\ref{lbl:rev_occ}.
As a result, we obtain $\tr{t}\preceq \rev{\tr{t}}$ and
$\rev{\tr{t}}\preceq \tr{t}$. We note that although $\preceq$ is now
circular, there are no loops of forward transitions only, and all
loops involve forward transitions followed by their reverse versions
(in the inverse order).

The conflict relation for forward transitions or reverse
transitions is defined correspondingly as in Section 2.
There is an immediate conflict between $t$ and $t'$, written as $t\#_0
t'$, if $\preS{t} \cap \preS{t'}\neq \emptyset$.
When $t$ and $t'$ are forward transitions, then $t\#_0 t'$ is as in
Section~\ref{sec:unf}.
If $t$ is $\tr t$ and $t'$ is $\rev{\tr{t'}}$, then $t\#_0 t'$ is
$\preS{\tr{t}} \cap \preS{\rev{\tr{t'}}} \neq \emptyset$, which is
equivalent to $\postS{\tr{t}} \cap \postS{\tr{t'}} \neq \emptyset$,
that trivially holds in occurrence nets.
Hence, the immediate conflict relation is empty between reverse
transitions.
The \emph{conflict} relation $\#$ between $t$ and $t'$ is as in in
Section~\ref{sec:unf}: $ t\# t'$ if there are forward transitions $\tr
t_1, \tr t_1' \in T_O$ such that $\tr t_1 \preceq t$, $\tr t_1'
\preceq t'$, and $\tr t_1 \#_0 \tr t_1'$.

Having introduced the notions of causality and conflict, we can now
define concurrent transitions in $\rev O$.
We follow the definition in Section~\ref{sec:unf}: $t, t'$ are
\emph{concurrent}, written $t\ co\ t'$, if $t\neq t'$ and
$t\not\preceq t'$, $t'\not\preceq t$, and $\neg t \# t'$.

\begin{lem}\label{mixedcon}
Let $t$ and $t'$ be enabled coinitial forward or reverse
transitions of a reversible occurrence net.
Then $t\ co\ t'$ if and only if 
(a) $\neg (t\ \#_0 t')$ if $t,t'$ are forward transitions, 
(b) $t\not\preceq t'$ and $t'\not\preceq t$ if $t,t'$ are reverse
transitions, and
(c) $t'\not\preceq t$ and $\rev{t'}\not\preceq t$ if $t$ is a
forward transition and $t'$ is a reverse transition.
\end{lem}

\begin{proof}
$\Rightarrow$ part: This is obvious by the definition of $co$ and
  Lemma~\ref{con=noconflict}.

\noindent
$\Leftarrow$ part: We have three cases. If $t$ and $t'$ are forward
transitions, then we are done by Lemma~\ref{con=noconflict}.
If $t,t'$ are reverse transitions, then since they are initial and
enabled, they cannot be in conflict unless they are in an immediate
conflict.
However, an immediate conflict of two reverse transitions is
equivalent to a backward conflict of their forward versions, which is
not allowed in occurrence nets.
Hence $t\ co\ t'$ as long as they do not cause each other.
Finally, consider a forward transition $t$ and a reverse
transition $t'$.
Then $t\not\preceq t'$, one of the conditions for $t\ co\ t'$, is
equivalent to no backward conflict between $t$ and $\rev{t'}$, which
is guaranteed in occurrence nets.
No immediate conflict between $t$ and $t'$ is equivalent to $\rev{t'}
\not\preceq t$ and, with $t'\not\preceq t$ given, we obtain
$t\ co\ t'$.
\end{proof}

The next result, which follows by the definition of the concurrency
relation on forward and reverse transitions, will be helpful in the
next section when we consider consecutive concurrent transitions.
\begin{lem}\label{con-rev}
Let $t;t'$ be a trace of a reversible occurrence net. Then,
$t\ co\ t'$ if and only if $\rev{t}\ co\ t'$.
\end{lem}

Definition~\ref{lbl:rev_occ} and the discussion above allow us to show
that $\rev O$ is a conservative extension of $O$.

\begin{lem} 
\label{lem:conservative-occurrence}
$(O, m) \red (O,m')$ if and only if $(\rev O, m) \fred (\rev O,m')$.
\end{lem}

In general, a reversible occurrence net is not an occurrence net.
This is because adding reverse transitions introduces cycles and
backward conflict.
Consider $N_1$ in Figure~\ref{fig:nets}.  We notice that initially
$\tr t_1$ and $\tr t_2$ are in conflict.
Then, in $\rev N_1$ in Figure~\ref{fig:rev-nets}, the place $\p c$
that contains a token has two reverse transitions in its preset,
namely $\rev{\tr t_2}$ and $\rev{\tr t_3}$, hence there is a backward
conflict.

\section{Properties}\label{sec:prop}

We now study the properties of the reversible versions of occurrence
nets. We follow here the approach of Danos and Krivine \cite{rccs} and
work with a transition system defined in Section~3.
The aim is to show that reversibility in reversible occurrence nets
is \emph{causal-consistent} (Theorem~\ref{th:causal}).
Informally, this means that we can reverse a firing (or a transition)
in a computation of a concurrent system as long as all firings
(transitions) caused by the firing (transition) have been undone
first.
We shall need several useful properties over firings, transitions or
their sequences before we can prove Theorem~\ref{th:causal}.

An important property of a \emph{fully} reversible system is the lemma
below stating that any forward transition can be undone.
In the setting of reversible occurrence nets this is stated as
follows:

\begin{lem}[Loop Lemma]\label{lm:loop}
Let $O$ be an occurrence net. Then, $(\rev O, m) \fred (\rev O,m')$ if
and only if $(\rev O, m') \bred[\rev {\tr t}] (\rev O,m)$.
\end{lem}

We can generalise this property to sequences of transitions and
reverse transitions as follows:

\begin{cor}\label{lb:corollary}
Let $O$ be an occurrence net. Then, $(\rev O, m) \red[]^{*} (\rev
O,m')$ if and only if $(\rev O, m') \red[]^{*} (\rev O,m)$.
\end{cor}

Next, we have a lemma which is instrumental in the proof of
causal-consistent reversibility~\cite{rccs,rhotcs}.
Note that $t$ and $t'$ can be either forward or reverse transitions.

\begin{lem}[Square Lemma]\label{lem:SL-all}
Let\; $t$ and $t'$ be enabled coinitial concurrent transitions of a
reversible occurrence net. Then, there exist transitions $t_1$ and
$t_1'$ such that $t;t_1'$ and $t';t_1$ are cofinal.
\end{lem}

\begin{proof}
We consider three cases when $t$ and $t'$ are enabled, coinitial and
concurrent transitions.
If $t$ and $t'$ are both forward then we are done by
Lemma~\ref{lem:SL-F}.
If $t$ and $t'$ are both reverse, then they cannot cause one another
and they cannot be in an immediate conflict (because this would imply
a backward conflict on their forward versions which we do not have in
occurrence nets: see Lemma~\ref{mixedcon}).
So, $t, t'$ do not share tokens in their presets. Hence, we can write
$t, t'$ correspondingly as we wrote $\tr t, \tr t'$ in the proof of
Lemma~\ref{lem:SL-F}.
The rest of the case follows correspondingly as in the proof of
Lemma~\ref{lem:SL-F}.
If $t$ is a transition and $t'$ is a reverse transition, then since
they are coinitial, concurrent and enabled they cannot cause one
another and they are not in an immediate conflict, which is equivalent
to $\rev{t'}\not\preceq t$.
Hence the preset of $t$ does not overlap with the preset of $t'$. We
then finish as in the proof of Lemma~\ref{lem:SL-F}.
\end{proof}

In order to prove causal consistency we first define a notion of
equivalence on sequences of forward and reverse transitions in
reversible occurrence nets.  By following the approach
in~\cite{levy,rccs}, we define the notion of {\em reverse equivalence}
on such sequences as the least equivalence relation $\ceq$ which is
closed under composition with `$;$' such that the following rules
hold:

\begin{align*}
t ; t' \ceq t' ; t  \;\; \mbox{ if } t\ co\ t'
\qquad\qquad
t ; \rev t \ceq \voidt
\qquad\quad
\rev{t} ; {t} \ceq \voidt
\end{align*}

The reverse equivalence $\ceq$ allows us to swap the order of $t$ and
$t'$ in an execution sequence as long as $t$ and $t'$ are concurrent.
Moreover, it allows cancellation of a transition and its inverse.  We
have that ${\equiv} \subset {\ceq}$.
The equivalence classes of $\ceq$ are called
\emph{traces}; it is clear that they contain the Mazurkiewicz traces.
Hence, we shall use $\omega, \omega_1$ and $\omega_2$ to range over
such traces.

The following lemma says that, up to reverse equivalence, one can
always reach for the maximum freedom of choice, going backwards first
and only then going forwards.

\begin{lem}[Parabolic Lemma]\label{lm:rearranging}
Let $\trace$ be a trace of a reversible occurrence net.  There
exist two forward traces $\trace_1$ and $\trace_2$ such that
$\trace \ceq \rev{\trace_1};\trace_2$.
\end{lem}

\begin{proof}
The proof is by lexicographic induction on the length of $\trace$ and
on the distance between the beginning of $\trace$ and
the \emph{earliest pair} of transitions in $\trace$ of the form
$\tr{t_1};\rev{\tr{t_2}}$, where $\tr{t_1}$ and $\tr{t_2}$ are
forward.
If such pair does not exist, then the result follows
immediately (i.e., either $\trace_1 = \epsilon$ or $\trace_2
= \epsilon$).
If there exists one such pair we have two possibilities: either (i)
$\tr{t_1}$ and $\rev{\tr{t_2}}$ are concurrent or (ii) they are not.
Case (i) also implies that $\rev{\tr{t_1}}$ and $\rev{\tr{t_2}}$ are
coinitial, enabled and also concurrent by Lemma~\ref{con-rev}.
Lemma~\ref{lem:SL-all} implies that there two cofinal firings with
labels $\rev{\tr{t_2}}$ and $\rev{\tr{t_1}}$ that making up a
square.
In addition to the subtrace $\tr{t_1}; \rev{\tr{t_2}}$, there is a
coinitial and cofinal subtrace $\rev{\tr{t_2}}; \tr{t_1}$ on the other
side of the square.
So $\tr{t_1}$ and $\rev{\tr{t_2}}$ can be swapped, obtaining a later
earliest pair of the form $\tr{t};\rev{\tr{t'}}$.
The result then follows by induction since swapping keeps the total
length of the trace unchanged.

In case (ii), $\tr{t_1}$ and $\rev{\tr{t_2}}$ are not
concurrent. Then, we have the following four cases:

\begin{enumerate}
\item $\tr{t_1}\preceq \rev{\tr{t_2}}$.
This implies that
$\postS{\tr{t_1}} \cap \preS{\rev{\tr{t_2}}} \neq \emptyset$ which is
equivalent to $\postS{\tr{t_1}}\cap \postS{\tr{t_2}} \neq \emptyset$.
This means that there is a backward conflict between $\tr{t_1}$ and
$\tr{t_2}$. Since occurrence nets are free of backward conflict we get
a contradiction.

\item $\rev{\tr{t_2}} \preceq {\tr{t_1}}$. 
This implies that
$\postS{\rev{\tr{t_2}}}\cap \preS{\tr{t_1}} \neq \emptyset$, and that
$\preS{\tr{t_2}}\cap \preS{\tr{t_1}}\neq \emptyset$.
The last says that $\tr{t_1}$ and $\tr{t_2}$ are in an immediate
conflict.
Since we have $\tr{t_1};\rev{\tr{t_2}}$ in $\omega$, and all
transitions before $\tr{t_1}$ in $\omega$ are forward transitions, it
means that to undo $\tr{t_2}$ the transition $\tr{t_2}$ must have
occurred before $\tr{t_1}$\footnote{Note that it is not possible that
some transition $\tr{t_3}$ has produced a token to a place $\p a$ that
was used by $\rev{\tr{t_2}}$: this would imply
$\{\tr{t_2}, \tr{t_3}\}=\preS{\p a}$, which is a backward conflict:
contradiction.}.
However, this contradicts $\tr{t_1}$ and $\tr{t_2}$ being in an
immediate conflict, which requires that if one transition takes place
then the other cannot happen as they share tokens in their preset and
there are no cycles among the forward transitions.

\item $\tr{t_1}\#_0 \rev{\tr{t_2}}$.
By definition of $\#_0$ we have
$\preS{\tr{t_1}} \cap \preS{\rev{\tr{t_2}}}\neq \emptyset$. Since
$\tr{t_1};\rev{\tr{t_2}}$ appears in $\omega$, $\rev{\tr{t_2}}$ should
remain enabled after the firing of $\tr{t_1}$.
Since occurrence nets are 1-safe and acyclic, this contradict the
assumption $\tr{t_1}\#_0 \rev{\tr{t_2}}$.

\item $\tr{t_1}\# \rev{\tr{t_2}}$.
This implies that there exist earlier $\tr{t}$ and $\tr{t'}$ in
$\omega$ such that $\tr{t} \preceq \tr{t_1}$ and
$\tr{t'} \preceq \rev{\tr{t_2}}$ with $\tr{t} \#_0 \tr{t'}$. So we
have $\tr{t}$ and $\tr{t'}$ in an immediate conflict appearing in
$\omega$ prior to sub-trace $\tr{t_1};\rev{\tr{t_2}}$.
We then show a contradiction as in case (3) above.
\end{enumerate}
\end{proof}

The following lemma says that, if two traces $\trace_1$ and $\trace_2$
are coinitial and cofinal (namely they start from the same marking and
end in the same marking) and if $\trace_2$ has only forward
transitions, then $\trace_1$ has some forward transitions and their
reverse versions that can cancel each other out.
A consequence of this is that overall $\trace_1$ is causally
equivalent to a forward trace $\trace'_1$ in which all pairs of
inverse transitions are cancelled out.

\begin{lem}[Shortening Lemma]\label{lm:short}
Assume $\trace_1$ and $\trace_2$ are coinitial and cofinal traces, and
$\trace_2$ is forward.
Then, there exists a forward trace $\trace'_1$ with
$\len{\trace'_1} \leq \len{\trace_1}$ such that
$\trace'_1 \ceq \trace_1$.
\end{lem}

\begin{proof}
By induction on the length of $\trace_1$. If $\trace_1$ is a forward
trace, then we are already done.
Otherwise, by applying Lemma~\ref{lm:rearranging}, we have that
$\trace_1 \ceq \rev\trace;\trace'$ where both $\trace$ and $\trace'$
are \textit{forward} traces.
Consider the sub-trace $\rev{\tr{t_1}};\tr{t_2}$ in $\rev\trace;\trace'$. 
Clearly $\rev{\tr{t_1}};\tr{t_2}$  is the only sub-trace of this form 
in $\rev\trace;\trace'$.
Suppose $\preS{\rev{\tr{t_1}}} = m'$ and $\postS{\rev{\tr{t_1}}} = m''$.  
Since $\trace_2$ is a forward only trace and since $\trace_1
(\ceq \rev\trace;\trace')$ and $\trace_2$ coinitial and cofinal, the
marking $m'$ removed by $\rev{\tr{t_1}}$ has to be put back in
$\trace_1$, otherwise this change would be visible in $\trace_2$
(since it is a forward only trace). Let $\tr{t_3}$ the earliest
transition in $\trace'$ able to consume $m''$ and put back $m'$.
This implies that $\tr{t_3} = \tr{t_1}$ since there are no loops of forward
transitions in occurrence nets.
If there are no forward transitions between $\rev{\tr{t_1}}$ and
$\tr{t_1}$, then we can cancel them out by applying the rules for
$\ceq$.
Then, we are done by induction on a shorter trace.

Otherwise, assume there are forward transitions between
$\rev{\tr{t_1}}$ and $\tr{t_1}$, and $\tr{t_4}$ is the last such
transition.
Next, we show that $\tr{t_4}$ and $\tr{t_1}$ are concurrent so that
they can be swapped thus moving $\tr{t_1}$ closer to $\rev{\tr{t_1}}$.
It is sufficient to show $\rev{\tr{t_4}}$ and $\tr{t_1}$ are
concurrent by Lemma~\ref{con-rev}.
So we require $\rev{\tr{t_4}}\not\preceq \tr{t_1}$ and
$\tr{t_4}\not\preceq \tr{t_1}$ by Lemma~\ref{mixedcon}. Since
$\tr{t_4};\tr{t_1}$ is a trace, the transitions cannot be in an
immediate conflict.
This is equivalent to
$\postS{\tr{t_4}} \cap \preS{\tr{t_1}}= \emptyset$, which is
$\rev{\tr{t_4}}\not\preceq \tr{t_1}$.
So, we are only left to prove $\tr{t_4}\not\preceq \tr{t_1}$. We
proceed by contradiction.  Assume that $\tr{t_4}\preceq \tr{t_1}$.
Hence, $\postS{\tr{t_4}} \cap \preS{\tr{t_1}}\neq \emptyset$, so $\p
a \in \postS{\tr{t_4}} \cap \preS{\tr{t_1}}$ for some place $\p a$.
After $\rev{\tr{t_1}}$ takes place in $\trace_1$ there will be a token
in $\p a$.
If that token stays there while computation progresses towards
$\tr{t_4}$, then $\tr{t_4}$ will place another token in $\p a$:
contradiction with the 1-safe property.
So, there is a transition $\tr{t_5}$ between $\rev{\tr{t_1}}$ and
$\tr{t_4}$ that consumes the token from $\p a$.
Hence, $\tr{t_5}$ and $ \tr{t_1}$ are in an immediate conflict and
form a forward trace $\tr{t_5}; \ldots;\tr{t_4};\tr{t_1}$:
contradiction.
 
Since $\rev{\tr{t_4}}$ and $\tr{t_1}$ are concurrent, so are
$\tr{t_4}$ and $\tr{t_1}$. Hence, $\tr{t_4};\tr{t_1}
\ceq \tr{t_1};\tr{t_4}$, which moves $\tr{t_1}$ closer to $\rev{\tr{t_1}}$.
We then continue this way until transitions $\rev{\tr{t_1}}$ and
$\tr{t_1}$ are adjacent.
Finally, we cancel them out by applying the rules for $\ceq$ and then
we conclude by induction on a shorter trace.
\end{proof}

Next, we give our causal consistency result. If two traces composed of
transitions or reverse transitions are coinitial and cofinal, then
they are equivalent with respect to $\ceq$, and vice versa.
Given a computation represented by trace $\trace$ we can reverse it by
simply doing $\rev{\trace}$, which would be \emph{backtracking}, or by
doing $\rev{\trace'}$ if $\trace \ceq \trace'$.
The latter option allows us to undo transitions in any order as long
as all the consequences of these transitions have been undone first.

\begin{thm}[Causal Consistency]\label{th:causal}
Let $\trace_1$ and $\trace_2$ be two coinitial traces. Then,
$\trace_1 \ceq \trace_2$ if and only if $\trace_1$ and $\trace_2$ are
cofinal.
\end{thm}

\begin{proof}
\noindent $\Rightarrow$ part: If $\trace_1$ and $\trace_2$ are coinitial
and $\trace_1 \ceq \trace_2$, then $\trace_1$ and $\trace_2$ are
cofinal.
We notice that if $\trace_1\ceq \trace_2$, then $\trace_1$ can be
transformed into $\trace_2$ (and vice versa) via $n \geq 0$
applications of the rules of $\ceq$.
So we proceed by induction on $n$. For $n = 0$ we have that
$\trace_1 \ceq \trace_2$ by applying $0$ times the rules of $\ceq$.
Since $\ceq$ is an equivalence, this means that $\trace_1 = \trace_2$
which in turn implies that the traces are coinitial and cofinal.
In the inductive case there exist $n$ traces $\trace^k$ (with $0 \leq
k \leq n$) obtained as a result of applying the rules of $\ceq$ to
$\trace_1$ exactly $k$ times; hence, $\trace^0 = \trace_1$ and
$\trace^n = \trace_2$.
We then have $\trace_1 \ceq \trace^{n-1}$, and
$\trace^{n-1} \ceq \trace_2$, i.e., traces $\trace^{n-1}$ and
$\trace_2$ differ in one rule application.
This means that we can decompose both traces as
$\trace^{n-1}= \trace_a;\trace';\trace_b$ and
$\trace_2= \trace_a;\trace'';\trace_b$, where $\trace'$ and $\trace''$
are shortest traces that differ by just one application of the rules
for $\ceq$.
Then, there are three cases:

\begin{enumerate}
\item
$\trace' = \tr{t_1};\tr{t_2}$ and $\trace'' = \tr{t_2};\tr{t_1}$ with
$\tr{t_1}\ co \ \tr{t_2}$;

\item $\trace' = \tr{t}; \rev{\tr{t}}$ and $\trace'' =\voidt_m$ with
$m$ the marking before firing $\tr{t}$; \item $\trace' = \rev{\tr{t}}
; \tr{t}$ and $\trace'' = \voidt_m$ with $m$ the marking after firing
$\tr{t}$.
\end{enumerate}

We need to show that $\trace_1$ and $\trace_2$ are cofinal.  In all
the cases above it is easy to see that $\trace^{n-1}$ and $\trace_2$
are both coinitial and cofinal.
By the inductive hypothesis, $\trace_1 \ceq \trace^{n-1}$ implies that
$\trace_1$ and $\trace^{n-1}$ are coinitial and cofinal.
Hence, we can conclude that also $\trace_1$ and $\trace_2$ are
coinitial and cofinal.

\vspace{.4cm}
\noindent $\Leftarrow$ part: 
If $\trace_1$ and $\trace_2$ are coinitial and cofinal, then
$\trace_1 \ceq \trace_2$.
We can assume by Lemma~\ref{lm:rearranging} that $\trace_1$ and
$\trace_2$ are compositions of a backward trace and a forward trace.
The proof is by lexicographic induction on the sum of the lengths of
$\trace_1$ and $\trace_2$, and on the distance between the end of
$\trace_1$ and $t_1$ of the earliest pair of $t_1$ in $\trace_1$ and
$t_2$ in $\trace_2$ such that $t_1 \neq t_2$. If all the transitions
in $\trace_1$ and $\trace_2$ are equal then we are done.
Otherwise, we have to consider three cases depending on the direction
of the two transitions.

\begin{enumerate}
\item $t_1=\tr{t_1}$ and $t_2=\rev{\tr{t_2}}$.
We have $\trace_1 = \rev\trace;\tr{t_1};\trace'$ and $\trace_2
= \rev\trace;\rev{\tr{t_2}};\trace''$, where $\rev\trace$ is the
common backward sub-trace and $\tr{t_1};\trace'$ is a forward trace.
Since $\trace_1$ and $\trace_2$ are coinitial and cofinal, this
implies that also $\tr{t_1};\trace'$ and $\rev{\tr{t_2}};\trace''$ are
coinitial and cofinal.
By applying Lemma~\ref{lm:short} to the trace
$\rev{\tr{t_2}};\trace''$ we obtain a shorter equivalent forward
trace, and therefore $\trace_2$ has a shorter causally equivalent
trace. We then conclude by induction.
\item $t_1=\tr{t_1}$  and $t_2=\tr{t_2}$.
By assumption we have $\tr{t_1} \neq \tr{t_2}$. We can decompose the
two traces as follows: $\trace_1 = \trace;\tr{t_1};\trace^1$ and
$\trace_2 = \trace;\tr{t_2};\trace^2$, for some $\trace, \trace^1$ and
$\trace^2$, where all the transitions in $\trace^1$ and $\trace^2$ are
forward.
Next, we show that $\tr{t_1}$ and $\tr{t_2}$ are concurrent. Assume
for contradiction that $\tr{t_1}$ and $\tr{t_2}$ are in an immediate
conflict. This means
$\preS{\tr{t_1}} \cap \preS{\tr{t_2}} \neq \emptyset$.
Since the two traces are cofinal we can write
\begin{align*}
  &
  (\rev{O},m_0) \red[\trace] (\rev{O},m) \fred[\tr{t_1}] (\rev{O},m_1)
  \fred[\trace^1] (\rev{O},m_f)
  \\
  &
  (\rev{O},m_0) \red[\trace] (\rev{O},m) \fred[\tr{t_2}] (\rev{O},m_2)
  \fred[\trace^2] (\rev{O},m_f)
\end{align*}
where $m_1 \neq m_2$ and $m_f$ is the final marking. The effect of
$\tr{t_1}$ remains visible at the end of $\trace^1$, namely in $m_f$,
and hence at the end of $\trace^2$.
This implies that there is a transition in $\trace^2$ that produces
the tokens necessary to enable $\tr{t_1}$.
Hence, there is a loop.  This is a contradiction since here we
consider forward transitions only of occurrence nets, which by
definition are acyclic with regards to forward transitions.
Hence, $\tr{t_1}\ co \ \tr{t_2}$ by Lemma~\ref{mixedcon}, and
$\tr{t_1}\in \trace^2$.
If there are no transitions between $\tr{t_2}$ and $\tr{t_1}$ in
$\trace^2$, then we swap them by applying the rules for $\ceq$ thus
obtaining $\trace_2\ceq \trace; \tr{t_1}; \tr{t_2}; \trace^3$, where
$\trace^3$ is $\trace^2$ without the leading $\tr{t_1}$.
The traces $\trace;\tr{t_1};\trace^1$ and
$\trace; \tr{t_1}; \tr{t_2}; \trace^3$ have the same length as before,
but the first pair of transitions that do not agree is closer to the
end of the trace, and the result follows by induction.
Otherwise, assume there are forward transitions between $\tr{t_2}$ and
$\tr{t_1}$, and let $\tr{t_3}$ is the last such transition.
We now show that $\tr{t_1}$ and $\tr{t_3}$ are concurrent, which
requires $\rev{\tr{t_3}}\not\preceq \tr{t_1}$ and
$\tr{t_3}\not\preceq \tr{t_1}$ by Lemma~\ref{con-rev} and
Lemma~\ref{mixedcon}.
Since $\tr{t_3};\tr{t_1}$ is a trace, the transitions cannot be in an
immediate conflict:
$\preS{\tr{t_3}} \cap \preS{\tr{t_1}}= \emptyset$.
This is equivalent to
$\postS{\rev{\tr{t_3}}} \cap \preS{\tr{t_1}}= \emptyset$,
which is $\rev{\tr{t_3}}\not\preceq \tr{t_1}$.
So we only need to prove $\tr{t_3}\not\preceq \tr{t_1}$. Assume for
contradiction that $\tr{t_3}\preceq \tr{t_1}$.
Hence, $\postS{\tr{t_3}} \cap \preS{\tr{t_1}}\neq \emptyset$, so $\p
a \in \postS{\tr{t_3}} \cap \preS{\tr{t_1}}$ for some place $\p a$.
Since $\tr{t_1}$ takes place in $\trace_1$ there is a token in $\p
a$. If that token stays there while computation progresses in
$\trace_2$ towards $\tr{t_3}$, then $\tr{t_3}$ will place another
token in $\p a$: contradiction with the 1-safe property.
So there is a transition $\tr{t_4}$ between $\tr{t_2}$ and $\tr{t_3}$
that consumes the token from $\p a$.
Hence, $\tr{t_4}$ and $ \tr{t_1}$ are in an immediate conflict and
form a forward trace $\tr{t_4};\ldots;\tr{t_3};\tr{t_1}$:
contradiction.
  
Since $\tr{t_3}$ and $\tr{t_1}$ are concurrent we can apply the rule
$\tr{t_3};\tr{t_1} \ceq \tr{t_1};\tr{t_3}$, which `moves' $\tr{t_1}$
closer to $\tr{t_2}$.
We then continue this way until $\tr{t_2}$ and $\tr{t_1}$ are adjacent
and are shown to be concurrent.
Finally, we swap them and, since we get a later pair of transitions
that do not agree, we conclude by induction.

\item $t_1=\rev{\tr{t_1}}$ and $t_2=\rev{\tr{t_2}}$. 
We have assumed that $t_1 \neq t_2$.
Moreover, the two transitions consume different tokens, that is
$\preS{\rev{\tr{t_1}}} \cap \preS{\rev{\tr{t_2}}} = \emptyset$,
otherwise we would have that
$\postS{\tr{t_1}} \cap \postS{\tr{t_2}} \neq \emptyset$, which is
impossible since occurrence nets have no backward conflicts.
Since the two traces are coinitial and cofinal, then either there
exists $\rev{\tr{t_1}} \in \trace_2$ or the effect of $\rev{\tr{t_1}}$
is undone later in the trace $\trace_1$ by $\tr{t_1}$.
We now consider the two cases in turn:

\begin{enumerate}
\item We can decompose the two traces as
\[
\begin{array}{lll}
  \trace_1
  &
  =
  &
  \rev{\trace^a};\rev{\tr{t_1}};\rev{\trace^b_1};\trace^c_1
  \\[4pt]

  \trace_2
  &
  =
  &
  \rev{\trace^a};\rev{\tr{t_2}};\rev{\trace^b_{21}};
  \rev{\tr{t_1}}\rev{\trace^b_{22}};\trace^c_2
\end{array}
\]
where $\rev{\trace^b_{21}}, \rev{\trace^b_{22}}$ are (potentially
empty) reverse traces and $\trace^c_1$ and $\trace^c_2$ are forward
traces.
Assume that $\rev{\tr{t_3}}$ is the last transition in
$\rev{\trace^b_{21}}$, i.e., just before $\rev{\tr{t_1}}$.
Next, we show that $\rev{\tr{t_3}}$ and $\rev{\tr{t_1}}$ are
concurrent, which requires $\rev{\tr{t_1}}\not\preceq \tr{t_3}$ and
$\tr{t_1}\not\preceq \tr{t_3}$ by Lemma~\ref{con-rev} and
Lemma~\ref{mixedcon}.
Since $\rev{\tr{t_3}};\rev{\tr{t_1}}$ is a trace, the transitions
cannot be in an immediate conflict:
$\preS{\rev{\tr{t_3}}} \cap \preS{\rev{\tr{t_1}}}= \emptyset$. This is
equivalent to
$\preS{\rev{\tr{t_3}}} \cap \postS{\tr{t_1}}= \emptyset$, which is
$\tr{t_1} \not\preceq \rev{\tr{t_3}}$.
So, we only need to show $\tr{t_1}\not\preceq \tr{t_3}$. Assume for
contradiction that $\tr{t_1}\preceq \tr{t_3}$.
Hence, $\postS{\tr{t_1}} \cap \preS{\tr{t_3}}\neq \emptyset$, so $\p
a \in \postS{\tr{t_1}} \cap \preS{\tr{t_3}}$ for some place $\p a$.
Since $\rev{\tr{t_1}}$ takes place in $\trace_1$ there is a token in
$\p a$. If that token stays there while computation progresses in
$\trace_2$ towards $\rev{\tr{t_3}}$, then $\rev{\tr{t_3}}$ will place
another token in $\p a$: contradiction.
So there is a transition $\rev{\tr{t_4}}$ between $\rev{\tr{t_2}}$ and
$\rev{\tr{t_3}}$ that consumes the token from $\p a$. Hence, $\p
a \in \preS{\rev{\tr{t_4}}}$, which is $\p a \in \postS{\tr{t_4}}$.
Combined with $\p a \in \postS{\tr{t_1}}$ the last given backward
conflict: contradiction.
  
Since $\rev{\tr{t_3}}$ and $\rev{\tr{t_1}}$ are concurrent we can move
$\rev{\tr{t_1}}$ closer to $\rev{\tr{t_2}}$ by using the rules of
$\ceq$.
We then continue this way until $\rev{\tr{t_2}}$ and $\rev{\tr{t_1}}$
are adjacent and are shown concurrent. Finally, we swap them and,
since we get a later pair of transitions that do not agree, we
conclude by induction.
 
\item
We have that
$\trace_1= \rev{\trace^a_1};\rev{\tr{t_1}};\rev{\trace^b_1};\trace^c_1$,
where $\trace^c_1$ is forward only trace that contains $\tr{t_1}$.
We show that $\rev{\tr{t_1}}$ is concurrent with all the reverse
transitions in $\rev{\trace^b_1}$ in the corresponding way as in part
(a) above.
We then swap $\rev{\tr{t_1}}$ with all the reverse transitions that
follow it, thus making it the last such transition of the trace.
Next, we show that $\tr{t_1}$ is concurrent with all transitions in
$\trace^c_1$ that appear before $\tr{t_1}$.
This is done very much as in the proof of Shortening Lemma. This
allows us to swap $\tr{t_1}$ with all forward transitions preceding
it, resulting in the sub-trace $\rev{\tr{t_1}};\tr{t_1}$.
We can then apply the axiom $\rev{\tr{t_1}};\tr{t_1} \ceq \voidt_m$
with $m$ being the marking of the net before firing $\rev{\tr{t_1}}$.
This gives a shorter trace equivalent to $\trace_1$, and we finish by
induction.
\end{enumerate}
\end{enumerate}
\end{proof}

With Theorem~\ref{th:causal} we proved that the notion of causal
consistency characterises a space for admissible rollbacks which are:
(i) consistent in the sense that they do not lead to previously
unreachable configurations, and (ii) flexible enough to allow
rearranging of independent undo events.
This implies that starting from an initial marking, all the markings
reached by mixed computations are markings that could be reached by
performing only forward computations.
Hence, we have:
\begin{thm}\label{th:rev_sem}
Let $O$ be an occurrence net and $m_0$ an initial marking. Then,
\begin{align*}
	(\rev O,m_0) \red[]^* (\rev O, m) \iff (\rev O,m_0) \fred[]^*
	(\rev O, m).
\end{align*}
\end{thm}

\begin{proof}
The $\Leftarrow$ part follows trivially since ${\fred[]} \subset
{\red[]}$.  For the $\Rightarrow$ part, we have that $(\rev
O,m_0) \red[\trace] (\rev O, m)$ for some $\trace$.
By using Lemma~\ref{lm:rearranging}, we obtain
$\trace \ceq \rev{\trace_1};\trace_2$ for some forward traces
$\trace_1, \trace_2$.
Since ${\trace} \ceq {\rev{\trace_1};\trace_2}$,
Theorem~\ref{th:causal} gives us $(\rev
O,m_0) \bred[\trace_1] \fred[\trace_2] (\rev O, m)$. Let us note that
$m_0$ is the initial marking, and this implies that no backward
computations can take place from $(\rev O, m_0)$.
This means that $\trace_1 = \voidt_{m_0}$. Hence, we have
$\trace \ceq \trace_2$, where $\trace_2$ is a forward trace.
Finally, Theorem~\ref{th:causal} gives us $(O,m_0) \fred[\trace_2]
(O,m)$, which implies $(\rev O, m_0) \fred[\trace_2] (\rev O,m)$ as
required.
\end{proof}

\section{Reversing \textsc{p/t} nets} 

This section takes advantage of the classical unfolding construction
for \textsc{p/t} nets and the reversible semantics of occurrence nets
to add causally-consistent reversibility to \textsc{p/t} nets.

\begin{defi}
  Let $(N, m)$ be a marked \textsc{p/t} net and $\mathcal{U}\lbrack
  N,m \rbrack$ its unfolding. The reversible version of $(N, m)$,
  written $\rev{(N, m)}$, is $\rev{\mathcal{U}\lbrack N, m\rbrack}$.
\end{defi}

The following result states that a reversible net is a conservative
extension of its original version, namely reversibility does not
change the set of reachable markings.
The result is a direct consequence of
Lemma~\ref{lem:conservative-occurrence} and the fact
that morphisms preserve reductions~\cite[Theorem
3.1.5]{Winskel86}.

\begin{lem}\label{lm:corr}
  $(N, m) \red[]^* (N,m')$ if and only if $\rev{(N, m)} \fred[]^*
  (\rev O,m'')$ and $m' = f_s(m'')$, where
  $(f_s,f_t): \mathcal{U}\lbrack N,m\rbrack \rightarrow N$, defined
  such that $f_S(\p a, \_, \_) = \p a$ and $f_T(\tr t, \_) = \tr t$,
  is the folding morphism.
\end{lem}

\begin{proof}
Let $\mathcal{U}\lbrack N,m\rbrack = O$ and $m_0$ be such that $m =
f_s(m_0)$: namely $m_0$ is the set of the minimal places of $O$.
By the unfolding construction, $(N, m) \red[]^* (N,m')$ implies $(O,
m_0)\red[] ^* (O,m'')$ for an appropriate $m''$.
Since morphisms preserve reductions~\cite[Theorem 3.1.5]{Winskel86},
$(O, m_0)\red[] ^* (O,m'')$ implies $(N, m) \red[]^* (N,m')$.
Hence, $(N, m) \red[]^* (N,m')$ if and only if $(O, m_0)\red[] ^*
(O,m'')$.
Finally, we observe that $(O, m) \red (O,m')$ if and only if $(\rev O,
m) \fred (\rev O,m')$ by Lemma \ref{lem:conservative-occurrence}.
\end{proof}

We remark that the reversible version of a \textsc{p/t} net is defined
as the reversible version of an occurrence net, namely its
unfolding.
Consequently, all properties shown in the previous section apply to
the reversible semantics of
\textsc{p/t} nets.
In particular, Lemma~\ref{lm:corr} combined with
Theorem~\ref{th:rev_sem} ensures that all markings reachable by the
reversible semantics are just the reachable markings of the original
P/T net.
Formally:
\begin{thm}
  $(N, m) \red[]^* (N,m')$ if and only if $\rev{(N, m)} \red[]^* (\rev
  O,m'')$ and $m' = f_s(m'')$, where $(f_s,f_t): \mathcal{U}\lbrack
  N,m\rbrack \rightarrow N$, defined such that $f_S(\p a, \_, \_) = \p
  a$ and $f_T(\tr t, \_) = \tr t$, is the folding morphism.
\end{thm}
\begin{proof}
The $\Rightarrow$ part holds thanks to Lemma \ref{lm:corr} and because
$\fred[] \subset \red[]$.
For the $\Leftarrow$ part we proceed as follows.
Since $N$ is the product of the unfolding, we have that $\rev{(N,
m)} \red[]^* (\rev O,m'')$ implies $\rev{(N, m)} \fred[]^* (\rev
O,m'')$ by Theorem \ref{th:rev_sem}.
Then we apply Lemma~\ref{lm:corr} and obtain $(N, m) \red[]^* (N,m')$
as required.
\end{proof}

\section{Finite Representation of Reversible \textsc{p/t} Nets}

As seen in Figure~\ref{fig:RevN3}, the reversible version of a finite
net may be infinite.  In this section we show how to represent
reversible nets in a compact, finite way by using coloured Petri nets.
We assume infinite sets $\varset$ of variables and $\colourset$ of
colours, defined such that $\varset\subset\colourset$.
We use $x, y, \ldots$ to range over $\varset$, $c, d, \ldots$ to range
over $\colourset$ and $\p c, \p d, \ldots$ to range over
$\colourset\setminus\varset$.
For $c\in\colourset$, we write $\vars c$ for the set of variables in
$c$. With abuse of notation we write $\vars m$ for the set of
variables in a multiset $m\in\nat^{\places \times \colourset}$.
Let $\sigma:\varset\rightarrow\colourset$ be a partial function and
$c$ a colour (also, $m\in \nat^{\places\times\colourset}$), we write
$c\sigma$ (respectively $m\sigma$) for the simultaneous substitution
of each variable $x$ in $c$ (respectively $m$) by $\sigma(x)$.
  
\begin{defi}[\textsc{c-p/t} net]
\label{def:cptnet}
A \emph{coloured place/transition net} (\textsc{c-p/t} net) is a
4-tuple $N = (S_N, T_N, \preS{\_}_N, \postS{\_}_N)$, where
$S_N\subseteq \mathcal{P}$ is the (nonempty) set of places,
$T_N \subseteq \transitions$ is the set of transitions, and the
functions $\preS{\_}_N, \postS{\_}_N:
T_N\rightarrow \nat^{S_N\times\colourset}$ assign source and target to
each transition defined such that $\vars {\postS {\tr t}}
\subseteq \vars {\preS {\tr t}}$.
A marking of a \textsc{c-p/t} net $N$ is a multiset over
$S_N\times\colourset$ that does not contain variables, i.e.,
$m \in \nat^{S\times \colourset}$ and $\vars m = \emptyset$.
A marked \textsc{c-p/t} net is a pair $(N, m)$ where $N$ is
a \textsc{p/t} net and $m$ is a marking of $N$.
\end{defi}

\textsc{c-p/t} nets generalise \textsc{p/t} nets by extending markings to   
multisets of coloured tokens, and transitions to patterns that need to
be instantiated with appropriate colours for firing, as formally
stated by the firing rule below.
\[
\begin{array}{c}
\mathrule{coloured-firing}
           {\tr t = m \;\pntrans\;  m' \in T_N}
           {(N, m\sigma\oplus m'')
            \red
            (N, m'\sigma\oplus m'')}
\end{array}
\]
The firing of a transition $t = m \;\pntrans\; m' $ requires to
instantiate $m$ and $m'$ by substituting variables by colours, namely
the firing of $t$ consumes the instance $m\sigma$ of the preset $m$
and produces the instance $m'\sigma$ of the postset $m'$.
Note that the initial marking of a net does not contain variables by
Definition~\ref{def:cptnet}.
Hence, $ m\sigma\oplus m''$ does not contain any
variable. Consequently, $\sigma(\vars {\preS {\tr t}}) \cap \varset
= \emptyset$.
Moreover, Definition~\ref{def:cptnet} also requires $\vars {\postS
{\tr t}} \subseteq \vars {\preS {\tr t}}$. Hence, $m'\sigma\oplus m''$
does not contain any variable.

\begin{exa}
\begin{figure}
\subfigure[]{
\label{fig:cpn-example-1}
  $$
  \xymatrix@R=1.5pc@C=1pc{
  \drawplace\ar_x[dr]
  \nameplaceup {\p a}
  \POS[]+<-0.1pc,.25pc>\drop{\scriptscriptstyle{{\p c_1}}}
  \POS[]-<-0.1pc,.3pc>\drop{\scriptscriptstyle{{\p c_2}}}
  &
  &
  \drawplace\ar^y[dl]
  \nameplaceup {\p b}
  \POS[]\drop{\scriptscriptstyle{{\p c_3}}}
  \\
  &
  \drawtrans {\tr{t}} \ar_{(x,y)}[d]
  &
  \\
  & 
  \drawplace
  \nameplacedown {\p c}
  &
  }
$$
}
\hspace{2cm}
\subfigure[]{
\label{fig:cpn-example-2}
$$
  \xymatrix@R=1.5pc@C=1pc{
  &
  \drawplace\ar_{(\p c_1,x)}[dl]\ar^{(x, \p c_2)}[dr]
  \nameplaceup {\p a}
  \POS[]+<-0.15pc,.2pc>\drop{\scriptscriptstyle{{(\p c_1,}}}
  \POS[]-<-0.15pc,.25pc>\drop{\scriptscriptstyle{{\p c_2)}}}
  \\
  \drawtrans {\tr{t}_1} \ar_{(x,\p c_3)}[d]
  &&
  \drawtrans {\tr{t}_2} \ar_{\p c_4}[d]
  &
  \\
  \drawplace
  \nameplacedown {\p b}
  &&
  \drawplace
  \nameplacedown {\p c}
  }
$$
}
\caption{A simple \textsc{c-p/t} net} 
\label{fig:cpn-example}
\end{figure}
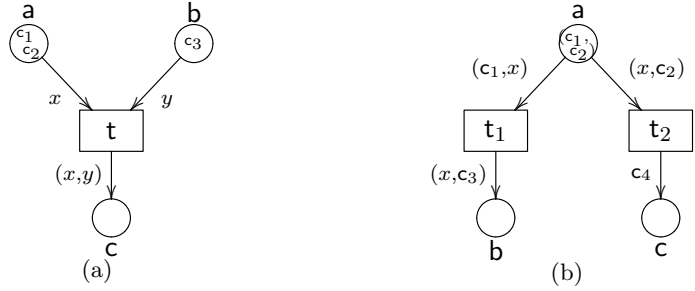

Consider the simple \textsc{c-p/t} net depicted in
Figure~\ref{fig:cpn-example-1}, which consists of the three places $\p
a$, $\p b$ and $\p c$ and the coloured transition $\tr t = \p
a(x) \oplus \p b (y) \;\pntrans\; \p c (x,y)$.
The firing of $\tr t$ consumes a token from $\p a$ and another one
from $\p b$ and produces a token in $\p c$, which is coloured by a
pair containing the colours of the consumed tokens.
Take the marking $m = \p a(\p c_1) \oplus \p a(\p c_2) \oplus \p b (\p
 c_3)$. There are two possible firings of $\tr t$ in $m$: one
 instantiates $x$ by $\p c_1$ and $y$ by $\p c_3$, and the other
 instantiates $x$ by $\p c_2$ and $y$ by $\p c_3$.
In the former case $m \red \p a(\p c_2)\oplus \p c (\p c_1,\p c_3)$;
in the latter $m \red \p a(\p c_1)\oplus \p c (\p c_2,\p c_3)$.

The enabling of a coloured transition may depend on the colour of the
tokens. For instance, the transition $\tr t_1$ in
Figure~\ref{fig:cpn-example-2} can be fired only when $\p a$ contains
a token coloured by a pair whose first component is $\p c_1$.
Similarly, $\tr t_2$ can be fired with a token from $\p a$ only when
its colour is a pair whose second component is $\p c_2$.
In particular, the marking $m = \p a (\p c_1, \p c_2)$ enables both
$\tr t_1$ and $\tr t_2$: in the former case $x$ is substituted by $\p
c_2$ and $m \red[\tr t_1] \p b(\p c_2,\p c_3)$; in the latter $x$ is
substituted by $\p c_1$ and $m \red[\tr t_2] \p c (\p c_4)$.
On the contrary, $\tr t_2$ is not enabled in $m' = \p a (\p c_1, \p
d)$ when $\p d \neq \p c_2$ because there available token $\p a(\p
c_1, \p d)$ is not an instance of the preset of $\tr t_2$, namely
there is no substitution $\sigma$ such that $\p a(\p c_1, \p d)= \p
a(x,\p c_2)\sigma$.
\end{exa}

We now exploit colours to attach to each token its execution history
and propose an encoding that associates each \textsc{p/t} net $N$ with
a \textsc{c-p/t} net $\enc N$ whose tokens carry their execution
history.
Our construction resembles the unfolding constructions of \textsc{p/t}
nets~\cite{NielsenPW81}.
Instead of using different places for representing tokens with
different causal histories and different transitions for representing
firings associated with different tokens, we use colour to distinguish
tokens and pattern matching in transitions to distinguish firings.

Our construction relies on the set of colours $\colourset$ defined as
the least set that contains $\varset$ and it is closed under the
following rules.
\[
\begin{array}{l}
  \mathrule{token}{
    h\in 2^{\colourset} \quad n\in\nat}{(h, n) \in \colourset
  }
  \hspace{1cm}
  \mathrule{elem}{
    x \in \transitions\cup\places\quad h\in 2^{\colourset}
  }{
    x(h) \in \colourset
  }
\end{array}
\]
Colours resemble the unfolding construction in
Figure~\ref{fig:unfolding}: the colours for tokens are $(h,n)$, where
$h$ denotes its (possibly empty) set of causes and $n$ is a natural
number used for distinguishing tokens with identical causal history.
Causal histories are built up from coloured versions of transitions
($\tr t(h)$) and places ($\p a(h)$).

\begin{defi}[P/T as C-P/T] \label{def:ptascpt}
Let $N = (S_N, T_N, \preS{\_}_N, \postS{\_}_N)$ be a \textsc{p/t} net.
Then, $\enc N$ is the \textsc{c-p/t} net defined as $\enc N = (S_N,
T_N, \preS{\_}_{\enc N}, \postS{\_}_{\enc N})$, where
5
\begin{itemize}
\item
$\preS{\tr t}_{\enc N} = \p{a_1}(x_1)\oplus\ldots\oplus\p{a_n}(x_n)$
where $\preS{\tr t}_{N} = \p{a_1}\ldots\p {a_n}$ and $\forall 1\leq
i,j \leq n. (x_i\in\varset \wedge x_i = x_j \implies i = j)$;

\item
      $\postS{\tr t}_{\enc N} =
          \{\p a (\{\tr t(h)\}, i) 
          \mid 
          \p a \in \supp{(\postS{\tr t}_{ N})}
          \  
          \wedge\  
          1 \leq i \leq \postS{\tr{t}}_N(\p a)
          \ 
          \wedge\  
         h = \preS{\tr t}_{\enc N}\}$.
\end{itemize}
A marked net $(N,m)$ is encoded as 
$\enc {(N,m)} =
              (\enc N, \enc m)$ where
$\enc m =  \{\p a (\emptyset, i) 
          \mid 
          \p a \in \supp{(m)}
          \  
          \wedge\  
          1 \leq i \leq m(\p a)
         \}$.
\end{defi}

The encoding does not alter the structure of a net, it only adds
colours to its tokens.
In fact, an encoded net has the same places and transitions as the
original net, and presets and postsets of each transition have the
same support.
Added colours do not interfere with firing because the preset of each
transition uses different colour variables for different tokens.
The colour $\{\tr t(h)\}$ assigned to each token produced by the
firing of $\tr t$ describes the causal history of the token, namely it
indicates that the token has been produced by $\tr t$ after consuming
the tokens in the preset of $\tr t$, which is denoted by $h$.
The natural numbers in the tokens in the postsets distinguish multiple
tokens produced by the same firing.
Tokens in the initial marking have empty causal history and are
coloured as $(\emptyset, i)$.

\begin{exa}
Consider the marked \textsc{p/t} net $(N,\p a\oplus\p a)$ in
Figure~\ref{fig:N-multiple}; its initial marking $\p a \oplus \p a$
assigns two tokens to place $\p a$ and each firing of its transition
$\tr {t_1}$ produces two token in place $\p b$.
Its encoding as a coloured \textsc{c-p/t} net is shown in
Figure~\ref{fig:N-multiple-cp}.
The colours of tokens in the initial marking have different natural
numbers as second components, and these numbers are inherited by the
two tokens generated by $\tr{t_1}$.

\begin{figure}[t]
\subfigure[$(N,\p a\oplus \p a)$]{ 
\label{fig:N-multiple}
$$
 \xymatrix@R=1pc@C=1pc{
 &
 \drawplace\ar@/_/[dl]
 \POS[]+<-0.2pc,.2pc>\drop{\bullet}
 \POS[]+<.2pc,-.2pc>\drop{\bullet}
 \nameplaceup{\p a}
 \\
 \drawtrans {\tr{t_1}} \ar@/_/_{2}[dr] 
 &
 &
 \drawtrans {\tr{t_2}} \ar@/_/[ul] 
 \\
 &
 \drawplace\ar@/_/[ur]
 \nameplacedown {\p b}
}$$
}
\hspace{2cm}
\subfigure[$\enc{(N,\p a\oplus \p a)}$]{ 
\label{fig:N-multiple-cp}
$$
 \xymatrix@R=1pc@C=1pc{
 &
 *[o]=<1.6pc,1.6pc>{\ }\drop\cir{}\ar@/_/_{x}[dl]
 \POS[]+<-0.1pc,.3pc>\drop{\scriptscriptstyle{(\!\emptyset,1\!)}}
 \POS[]+<.1pc,-.3pc>\drop{\scriptscriptstyle{(\!\emptyset,2\!)}}
 \POS[]+<0pc,1.2pc>\drop{{\p a}}
 \\
 \drawtrans {\tr{t_1}} \ar@/_/_{(\tr{t_1}(\p a(x)),1)\oplus (\tr{t_1}(\p a(x)),2)}[dr] 
 &
 &
 \drawtrans {\tr{t_2}} \ar@/_/_{(\tr{t_2}(\p b(x)),1)}[ul] 
 \\
 &
 \drawplace\ar@/_/_{x}[ur]
 \nameplacedown {\p b}
}
$$
}
\caption{A simple \textsc{P/T} with
multiple tokens with identical causal history}
\label{fig:multiple-tokens}
\end{figure}
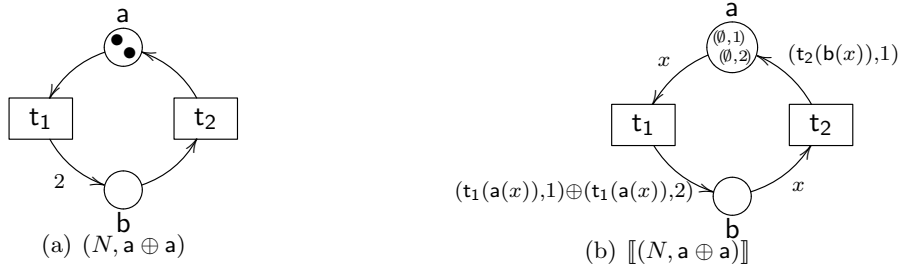
\end{exa}

\begin{exa}
\label{ex:encoding-nets}
The encoding of the nets in Figure~\ref{fig:nets} are shown in
Figure~\ref{fig:enc-nets}.
We comment on $\enc{N_1}$, which is the encoding of $N_1$.
The transition $\tr{t_1} = \p a \pntrans \p c$ in $N_1$ is encoded as
$\p a(x) \pntrans \p c(\tr{t_1}(a(x)),1)$, i.e., the firing of
$\tr{t_1}$ consumes a token with colour $h$ from place $\p a$ and
generates a token in $\p c$ with colour $(\tr{t_1}(a(h)),1)$.
The transition $\tr{t_2} = \p b \oplus c \pntrans \p e$ has two places
in the preset and uses two variables $x$ and $y$ in its encoded form
$\p b(x) \oplus c(y) \pntrans \p e(\tr{t_2}(\p b(x)\oplus \p
c(y)),1)$.
Note that the colour of the token produced in $\p c$ carries the
information of the tokens consumed from both places $\p b$ and $\p
c$.
The encoding for $\tr {t_3}$ is defined analogously.

We illustrate a sequence of firings of $\enc{(N_1, \p a \oplus \p
b \oplus \p c \oplus \p d)}$.
We remark that all the three transitions of the net $\enc{N_1}$ are
enabled at the initial marking.
Let us consider the case in which $\tr t_1$ is fired by consuming the
unique token in $\p a$:
\begin{align*}
  \enc{(N_1, \p a \oplus \p b \oplus \p c \oplus \p d)}  = \
  &
  (\enc{{N_1}}, \p a(\emptyset,1) \oplus \p b(\emptyset,1)
  \oplus
  \p c(\emptyset,1)\oplus \p d(\emptyset,1))
  \\
  \red [{\tr {t_1}}] 
  &
  (\enc{{N_1}}, \p b(\emptyset,1)
  \oplus \p c(\tr{t_1}(\p a(\emptyset,1)),1)
  \oplus \p c(\emptyset,1)
  \oplus \p d(\emptyset,1))
\end{align*}
In this case, $\tr t_1$ consumes a token of colour $(\emptyset,1)$
from $\p a$ and produces a token of colour $(\tr{t_1}(\p a$
$(\emptyset,1)),1)$ in $\p c$.
The causal history of the produced token, i.e., the first component of
its colour, $\tr{t_1}(\p a(\emptyset,1)$ indicates that the token has
been produced by the firing of $\tr t_1$ that consumed a token of
colour $(\emptyset,1)$ from $\p a$.
In the obtained marking both $\tr {t_2}$ and $\tr{t_3}$ are
enabled. Moreover, each transition can be fired in two different ways
depending on which one of the two available tokens in $\p c$ is
consumed.
One possible firing of $\tr{t_2}$ is 
\begin{align*}
 &
 (\enc{{N_1}}, \p b(\emptyset,1)
 \oplus \p c(\tr{t_1}(\p a(\emptyset,1)),1)
 \oplus \p c(\emptyset,1)
 \oplus \p d(\emptyset,1)) 
 \\
 &
 \hspace{3cm}\red[{\tr{t_2}}]
 (\enc{{N_1}}, \p e(\tr{t_2}(\p b(\emptyset,1)
 \oplus \p c(\emptyset,1),1))
 \oplus \p c(\tr{t_1}(\p a(\emptyset,1)),1)
 \oplus \p d(\emptyset,1)) 
\end{align*}
which consumes the tokens in $\p b$ and $\p c$ available in the
original marking.
Alternatively, 
\begin{align*}
  &
  (\enc{{N_1}}, \p b(\emptyset,1)
  \oplus \p c(\tr{t_1}(\p a(\emptyset,1)),1)
  \oplus \p c(\emptyset,1)
  \oplus \p d(\emptyset,1)) 
  \\
  &
  \hspace{3cm}\red[{\tr{t_2}}]
  (\enc{{N_1}}, \p e(\tr{t_2}(\p b(\emptyset,1)
  \oplus \p c(\tr{t_1}(\p a(\emptyset,1)),1)),1)
  \oplus \p c(\emptyset,1)
  \oplus \p d(\emptyset,1)) 
\end{align*}
where the firing of $\tr t_2$ consumes the token in $\p b$ generated
by the previous firing of $\tr t_1$.

We remark that nets in Figure~\ref{fig:enc-nets} use colours whose
second component is $1$.
We recall that natural numbers play the same r\^ole here as in the
standard definition of the unfolding, which is to distinguish tokens
that have the same causal history (see Section~\ref{sec:unf}).
All nets in Figure~\ref{fig:nets} have sets as initial markings and
transitions that do not generate multiple tokens in the same place;
consequently, the usage of natural numbers in these cases is
inessential.
\end{exa}

Although the second component in colours plays no role when reversing
nets, they allow us to state a tight correspondence between the
semantics of the coloured version of a \textsc{p/t} net and its
unfolding.

\begin{lem} 
\label{lemma:correspondence-unfolding-coloured}
Let $(N,m)$ be a marked \textsc{p/t} net and $\mathcal{U}\lbrack
N,m\rbrack = (O,m')$ its unfolding.
Then, $\enc{N,m} \red[s] (\enc{N},m'')$ if and only if $(O,m') \red[s]
(O,m'')$.
\end{lem}

\begin{proof} The $\Leftarrow$ part is by induction
on the length of the reduction. The base case follows by taking $m'' =
m'$ and noting that $\enc{N,m} = (\enc{N}, m')$.
The inductive step $s = s';\tr t$ follows by applying inductive
hypothesis on $s'$ to conclude that $\enc{N,m} \red[s']
(\enc{N},m''')$ iff $(O,m) \red[s'] (O,m''')$. {\em If)}
$(\enc{N},m''') \red(\enc{N},m'')$ implies $m''' = \preS{\tr
t}_{\enc{N}} \oplus m''''$ and $m'' = \postS {\tr t}_{\enc{N}} \oplus
m''''$.
Since $(O,m) \red[s'] (O,m''')$, $CO(\preS{\tr t})$. By the unfolding
construction we conclude $(O,m''') \red (O,m'')$. The $\Rightarrow$
part follows analogously.
\end{proof}

We now recast the notion of reversible \textsc{p/t} net by taking
advantage of the above correspondence.
In particular, we show that the reversible version of $\enc{N}$ can be
defined analogously to the case of occurrence nets, i.e., by adding
transitions that are the swapped versions of the ones in $N$.

\begin{defi}[Reversible \textsc{p/t} net] Let $N$ be a
\textsc{p/t} net. The reversible version of $N$ is $\rev{\enc N}$.
The reversible version of a marked \textsc{p/t} net $(N, m)$ is the
marked \textsc{c-p/t} net $(\rev{\enc N}, \enc m)$.
\end{defi}

We remark that $\rev{\enc N}$ is a \textsc{c-p/t} net. Note that the
definition of \textsc{c-p/t} net (Definition~\ref{def:cptnet}) imposes
$\vars {\postS {\tr t}} \subseteq \vars {\preS {\tr t}}$ for any
transition $\tr t$.
Consequently, the addition of reverse transitions to a \textsc{c-p/t}
net may not produce a \textsc{c-p/t} net.
However, the definition of $\enc N$ (Definition~\ref{def:ptascpt}) and
the fact that we consider nets without transitions with empty postsets
ensures $\vars {\postS {\tr t}} = \vars {\preS {\tr t}}$ for any $\tr
t\in\enc N$.

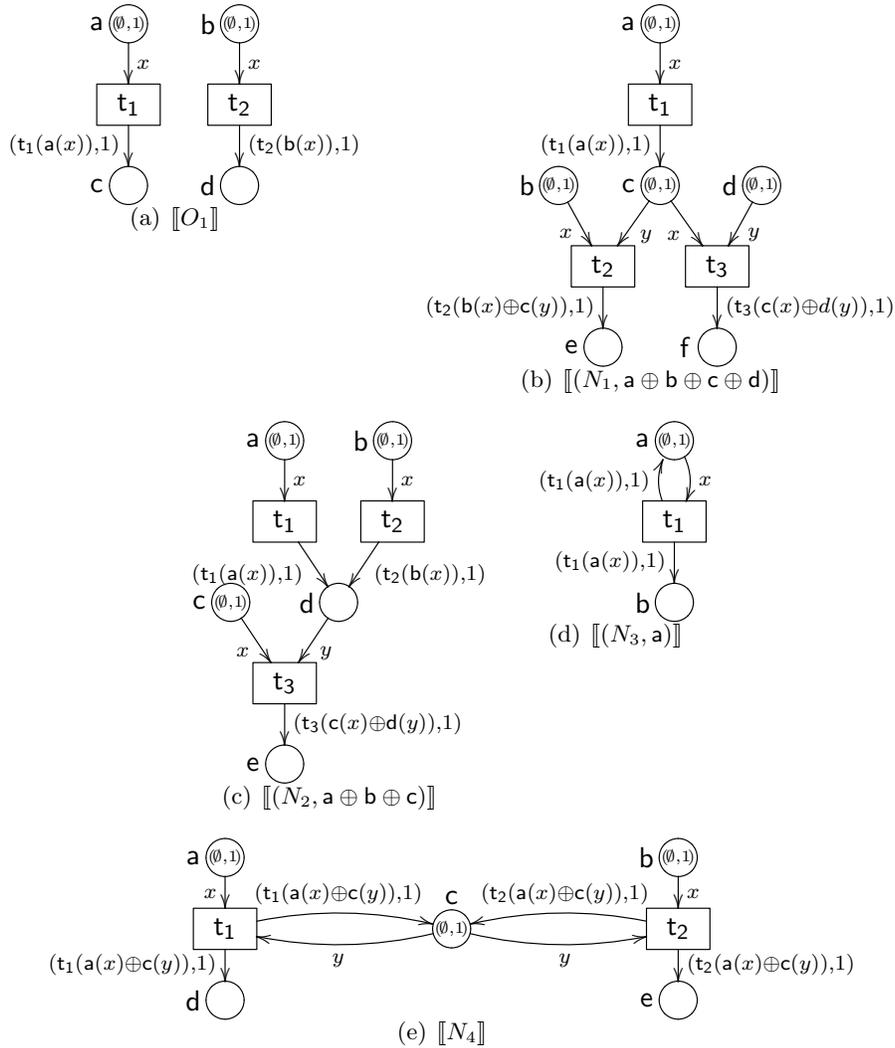
\begin{figure}[t]
\subfigure[$\enc{O_1}$]{ 
\label{fig:enc-O1}
$$
 \xymatrix@R=1.3pc@C=1.5pc{
 \drawcolouredmarkedplace{\scriptscriptstyle{(\!\emptyset,1\!)}}\ar^x[d]
 \nameplaceleft {\p a}
 &
 \drawcolouredmarkedplace{\scriptscriptstyle{(\!\emptyset,1\!)}}\ar^x[d]
 \nameplaceleft {\p b}
 \\
 \drawtrans {\tr{t_1}} \ar_{(\tr{t_1}(\p a(x)),1)}[d] 
 &
 \drawtrans {\tr{t_2}} \ar^{(\tr{t_2}(\p b(x)),1)}[d] 
 \\
 \drawplace
 \nameplaceleft {\p c}
 &
 \drawplace
 \nameplaceleft {\p d}
 }
$$
}
\subfigure[$\enc{(N_1, \p a \oplus \p b \oplus \p c \oplus \p d)}$]{ 
\label{fig:enc-N1}
$$
 \xymatrix@R=1.3pc@C=-.2pc{
 &
 &
 \drawcolouredmarkedplace{\scriptscriptstyle{(\!\emptyset,1\!)}}\ar^x[d]
 \nameplaceleft {\p a}
 \\
 &
 &
 \drawtrans {\tr{t_1}} \ar_{(\tr{t_1}(\p a(x)),1)}[d]
 \\
 \drawcolouredmarkedplace{\scriptscriptstyle{(\!\emptyset,1\!)}}\ar_x[dr]
 \nameplaceleft {\p b}
 &
 &
 \drawcolouredmarkedplace{\scriptscriptstyle{(\!\emptyset,1\!)}}\ar^y[dl]\ar_x[dr]
 \nameplaceleft {\p c}
 &
 &
 \drawcolouredmarkedplace{\scriptscriptstyle{(\!\emptyset,1\!)}}\ar^y[dl]
 \nameplaceleft {\p d}
 &
 &
 \\
 &
  \drawtrans {\tr{t_2}} \ar_{(\tr{t_2}(\p b(x)\oplus \p c(y)),1)}[d]
 &
 &
  \drawtrans {\tr{t_3}} \ar^{(\tr{t_3}(\p c(x)\p\oplus d(y)),1)}[d]
 \\
 &
 \drawplace
 \nameplaceleft {\p e}
 &
 &
  \drawplace
 \nameplaceleft {\p f}
 }
$$
}
\hspace{-1cm}
\subfigure[$\enc{(N_2, \p a \oplus \p b \oplus \p c)}$]{ 
\label{fig:enc-N2}
$$
 \xymatrix@R=1.3pc@C=.1pc{
 &
 \drawcolouredmarkedplace{\scriptscriptstyle{(\!\emptyset,1\!)}}\ar^x[d]
 \nameplaceleft {\p a}
 &
 &
 \drawcolouredmarkedplace{\scriptscriptstyle{(\!\emptyset,1\!)}}\ar^x[d]
 \nameplaceleft {\p b}
 \\
 &
 \drawtrans {\tr{t_1}} \ar_{(\tr{t_1}(\p a(x)),1)}[dr] 
 &
 &
 \drawtrans {\tr{t_2}} \ar^{(\tr{t_2}(\p b(x)),1)}[dl]  
 \\
\drawcolouredmarkedplace{\scriptscriptstyle{(\!\emptyset,1\!)}}\ar_x[dr]
 \nameplaceleft {\p c}
 &
 &
 \drawplace\ar^y[dl]
 \nameplaceleft {\p d}
 &
 &
 \\
 &
  \drawtrans {\tr{t_3}} \ar^{(\tr{t_3}(\p c(x)\oplus \p d(y)),1)}[d]
 &
 &
 \\
 &
 \drawplace
 \nameplaceleft {\p e}
 &
 }
$$
}
\hspace{-.3cm}
\subfigure[$\enc{(N_3,\p a)}$]{ 
\label{fig:enc-N3}
$$
  \xymatrix@R=1.3pc@C=.1pc{
 \drawcolouredmarkedplace{\scriptscriptstyle{(\!\emptyset,1\!)}}\ar^x@/^/[d]
 \nameplaceleft {\p a}
  \\
 \drawtrans {\tr{t_1}}  \ar_{(\tr{t_1}(\p a(x)),1)}[d]\ar^{(\tr{t_1}(\p a(x)),1)}@/^/[u]
  \\
 \drawplace
 \nameplaceleft {\p b}
 }
$$
}
\subfigure[$\enc{N_4}$]{ 
\label{fig:enc-N4}
$$
  \xymatrix@R=1pc@C=2.5pc{
 \drawcolouredmarkedplace{\scriptscriptstyle{(\!\emptyset,1\!)}}\ar_x[d]
 \nameplaceleft {\p a}
 &&
 &&
 \drawcolouredmarkedplace{\scriptscriptstyle{(\!\emptyset,1\!)}}\ar^x[d]
 \nameplaceleft {\p b}
 \\
 \drawtrans {\tr{t_1}} \ar_{(\tr{t_1}(\p a(x)\oplus\p c(y)),1)}[d] \ar^{(\tr{t_1}(\p a(x)\oplus\p c(y)),1)}@/^/[rr]
 &&
 \drawcolouredmarkedplace{\scriptscriptstyle{(\!\emptyset,1\!)}}\ar_y@/_/[rr]\ar^y@/^/[ll]
 \nameplaceup {\p c} 
 &&
 \drawtrans {\tr{t_2}} \ar^{(\tr{t_2}(\p a(x)\oplus\p c(y)),1)}[d] \ar_{(\tr{t_2}(\p a(x)\oplus\p c(y)),1)}@/_/[ll]
 \\
 \drawplace
 \nameplaceleft {\p d}
 &&
 &&
 \drawplace
 \nameplaceleft {\p e}
 }
$$
}
\caption{{\sc p/t} nets as {\sc c-p/t} nets}\label{fig:enc-nets}
\end{figure}

\begin{figure}[htp]
  \vspace{-5mm}
\subfigure[$\rev{\enc{O_1,\p a\oplus\p b}}$]{ 
\label{fig:REO1}
$$
 \xymatrix@R=1.3pc@C=1.5pc{
 \drawcolouredmarkedplace{\scriptscriptstyle{(\!\emptyset,1\!)}}\ar_x[d]
 \nameplaceleft {\p a}
 &&&
 &
 \drawcolouredmarkedplace{\scriptscriptstyle{(\!\emptyset,1\!)}}\ar_x[d]
 \nameplaceleft {\p b}
 \\
 \drawtrans {\tr{t_1}} \ar_{(\tr{t_1}(\p a(x)),1)}[d] 
 &
 \drawtrans {{\rev{\tr{ t_1}}}} \ar_{x}@{-->}@/_/[ul]  
 &&&
 \drawtrans {\tr{t_2}} \ar_{(\tr{t_2}(\p b(x)),1)}[d] 
 &
 \drawtrans {{\rev {\tr{t_2 }}}} \ar_{x}@{-->}@/_/[ul]
 \\
 \drawplace\ar_{({\tr{t_1}(\p a(x))},1)}@{-->}@/_/[ur]  
 \nameplaceleft {\p c}
 &&&
 &
 \drawplace\ar_{({\tr{t_2}(\p b(x))},1)}@{-->}@/_/[ur]  
 \nameplaceleft {\p d}
 }
$$
}
\hspace{-1.2cm}
\subfigure[$\rev{\enc{(N_1,\p a \oplus \p b\oplus \p c\oplus \p d)}}$]{ 
\label{fig:REN1}
$$
 \xymatrix@R=1.3pc@C=1.5pc{
 &
 &
 \drawcolouredmarkedplace{\scriptscriptstyle{(\!\emptyset,1\!)}}\ar_x[d]
 \nameplaceleft {\p a}
 \\
 &
 &
 \drawtrans {\tr{t_1}} \ar_{(\tr{t_1}(\p a(x)),1)}[d]
 &
 \drawtrans { \rev{\tr{ t_1}}} \ar_{x}@{-->}@/_/[ul]   
 \\
 \drawcolouredmarkedplace{\scriptscriptstyle{(\!\emptyset,1\!)}}\ar_x[dr]
 \nameplaceleft {\p b}
 &
 &
 \drawcolouredmarkedplace{\scriptscriptstyle{(\!\emptyset,1\!)}}\ar^y[dl]\ar_x[dr]
 \ar_{(\tr{t_1}(\p a(x)),1)}@{-->}@/_/[ur] 
 \nameplaceleft {\p c}
 &
 &
 \drawcolouredmarkedplace{\scriptscriptstyle{(\!\emptyset,1\!)}}\ar^y[dl]
 \nameplaceleft {\p d}
 &
 &
 \\
 &
  \drawtrans {\tr{t_2}} \ar_{(\tr{t_2}(\p b(x)\oplus\p c(y)),1)}[d]
 &
 &
  \drawtrans {\tr{t_3}} \ar^{(\tr{t_3}(\p c(x)\oplus\p d(y)),1)}[d]
 \\
 &
 \drawplace\ar_{(\tr{t_2}(\p b(x)\oplus\p c(y)),1)}@{-->}[d]
 \nameplaceleft {\p e}
 &
 &
  \drawplace\ar^{(\tr{t_3}(\p c(x)\oplus\p d(y)),1)}@{-->}[d]
 \nameplaceleft {\p f}
 \\
 &
  \drawtrans { \rev{\tr{ t_2}}} 
  \ar^{ x}@{-->}@/^14ex/[uuul] 
  \ar^{ y}@{-->}@/_3ex/[uuur]     
 &
 &
  \drawtrans{ \rev{\tr{ t_3}}} 
  \ar_{ x}@{-->}@/^3ex/[uuul] 
  \ar_{ y}@{-->}@/_14ex/[uuur]     
 }
$$
}
\subfigure[$\rev{\enc{(N_2, \p a \oplus\p b\oplus \p c)}}$]{ 
\label{fig:rev-enc-N2}
$$
 \xymatrix@R=1.3pc@C=1.5pc{
 &
 \drawcolouredmarkedplace{\scriptscriptstyle{(\!\emptyset,1\!)}}\ar^x[d]
 \nameplaceup {\p a}
 &
 &
 \drawcolouredmarkedplace{\scriptscriptstyle{(\!\emptyset,1\!)}}\ar^x[d]
 \nameplaceup {\p b}
 \\
 \drawtrans {\rev {\tr{t_1}}} 
  \ar^{x}@/^1ex/@{-->}[ur]
 &
 \drawtrans {\tr{t_1}} \ar^{(\tr{t_1}(\p a(x)),1)}[dr] 
 &
 &
  \drawtrans {\tr{t_2}} \ar^{(\tr{t_2}(\p b(x)),1)}[dl]  
  &
  \drawtrans {\rev {\tr{t_2}}}  
  \ar_{x}@/_1ex/@{-->}[ul]
 \\
\drawcolouredmarkedplace{\scriptscriptstyle{(\!\emptyset,1\!)}}\ar_x[drr]
 \nameplaceleft {\p c}
 &&
 \drawplace\ar^y[d]
 \ar_{(\tr{t_1}(\p a(x)),1)}@/^3ex/@{-->}[ull]
 \ar_{(\tr{t_2}(\p b(x)),1)}@/_3ex/@{-->}[urr]
 \nameplaceup {\p d}
 &
 &
 \\
 &
 &
  \drawtrans {\tr{t_3}} \ar_{(\tr{t_3}(\p c(x)\oplus\p d(y)),1)}[d]
 &
 &
 \\
 &&
 \drawplace\ar_{(\tr{t_3}(\p c(x)\oplus\p d(y)),1)}@{-->}[d]
 \nameplaceleft {\p e}
 &
 \\
 &&
 \drawtrans {\rev {\tr{t_3}}} 
  \ar_{y}@/_14ex/@{-->}[uuu]
  \ar^{x}@/^10ex/@{-->}[uuull]
 }
$$
}
\hspace{.1cm}
\subfigure[$\rev{\enc{(N_3, \p a)}}$]{ 
\label{fig:rev-enc-N3}
$$
  \xymatrix@R=1.3pc@C=2pc{
 \drawcolouredmarkedplace{\scriptscriptstyle{(\!\emptyset,1\!)}}\ar^x@/^/[d]
   \ar^{(\tr{t_1}(\p a(x)),1)}@{-->}@/^/[drr]
   \nameplaceleft {\p a}
 &&
  \\
 \drawtrans {\tr{t_1}}  \ar_{(\tr{t_1}(\p a(x)),1)}[d]
 \ar^{(\tr{t_1}(\p a(x)),1)}@/^/[u]
 &&
  \drawtrans {\rev {\tr{t_1}}}  
  \ar^{x}@{-->}@/^/[ull]
  \\
 \drawplace  \ar_{(\tr{t_1}(\p a(x)),1)}@{-->}@/_/[urr]
 \nameplaceleft {\p b}
 &
 &
 }
$$
}
\subfigure[$\rev{\enc{N_4, \p a \oplus\p b\oplus \p c}}$]{ 
\label{fig:rev-enc-N4}
$$
  \xymatrix@R=1.5pc@C=2.5pc{
 \drawcolouredmarkedplace{\scriptscriptstyle{(\!\emptyset,1\!)}}\ar_x[dr]
 \nameplaceup {\p a}
 &&&&
 &&
 \drawcolouredmarkedplace{\scriptscriptstyle{(\!\emptyset,1\!)}}\ar^x[dl]
 \nameplaceup {\p b}
 \\
 &\drawtrans {\tr{t_1}} \ar_{(\tr{t_1}(\p a(x)\oplus\p c(y)),1)}[d] \ar^{(\tr{t_1}(\p a(x)\oplus\p c(y)),1)}@/^/[rr]
 &&
 \drawcolouredmarkedplace{\scriptscriptstyle{(\!\emptyset,1\!)}}\ar_y@/_/[rr]\ar^y@/^/[ll]
 \ar^{(\tr{t_1}(\p a(x)\oplus\p c(y)),1)}@{-->}@/^/[ddll]
 \ar_(.7){(\tr{t_2}(\p b(x)\oplus\p c(y)),1)}@{-->}@/_/[ddrr]
 \nameplaceup {\p c} 
 &&
 \drawtrans {\tr{t_2}} \ar^{(\tr{t_2}(\p b(x)\oplus\p c(y)),1)}[d] \ar_{(\tr{t_2}(\p b(x)\oplus\p c(y)),1)}@/_/[ll]
 \\
 &
 \drawplace\ar_(.4){(\tr{t_1}(\p a(x)\oplus\p c(y)),1)}@{-->}[d]
 \nameplaceleft {\p d}
 &&
 &&
 \drawplace\ar^(.4){(\tr{t_2}(\p b(x)\oplus\p c(y)),1)}@{-->}[d]
 \nameplaceright {\p e}
 \\
 &
  \drawtrans{\rev{\tr{t_1}}}  \ar^{y}@/^/@{-->}[rruu]\ar^{x}@/^15ex/@{-->}[uuul]
&&
&&
 \drawtrans{\rev{\tr{t_2}}}  \ar_{y}@/_/@{-->}[lluu] \ar_{x}@/_15ex/@{-->}[uuur]
 &&
 }
$$
}
\caption{Reversible  coloured nets}\label{fig:rev-coloured-nets}
\end{figure}
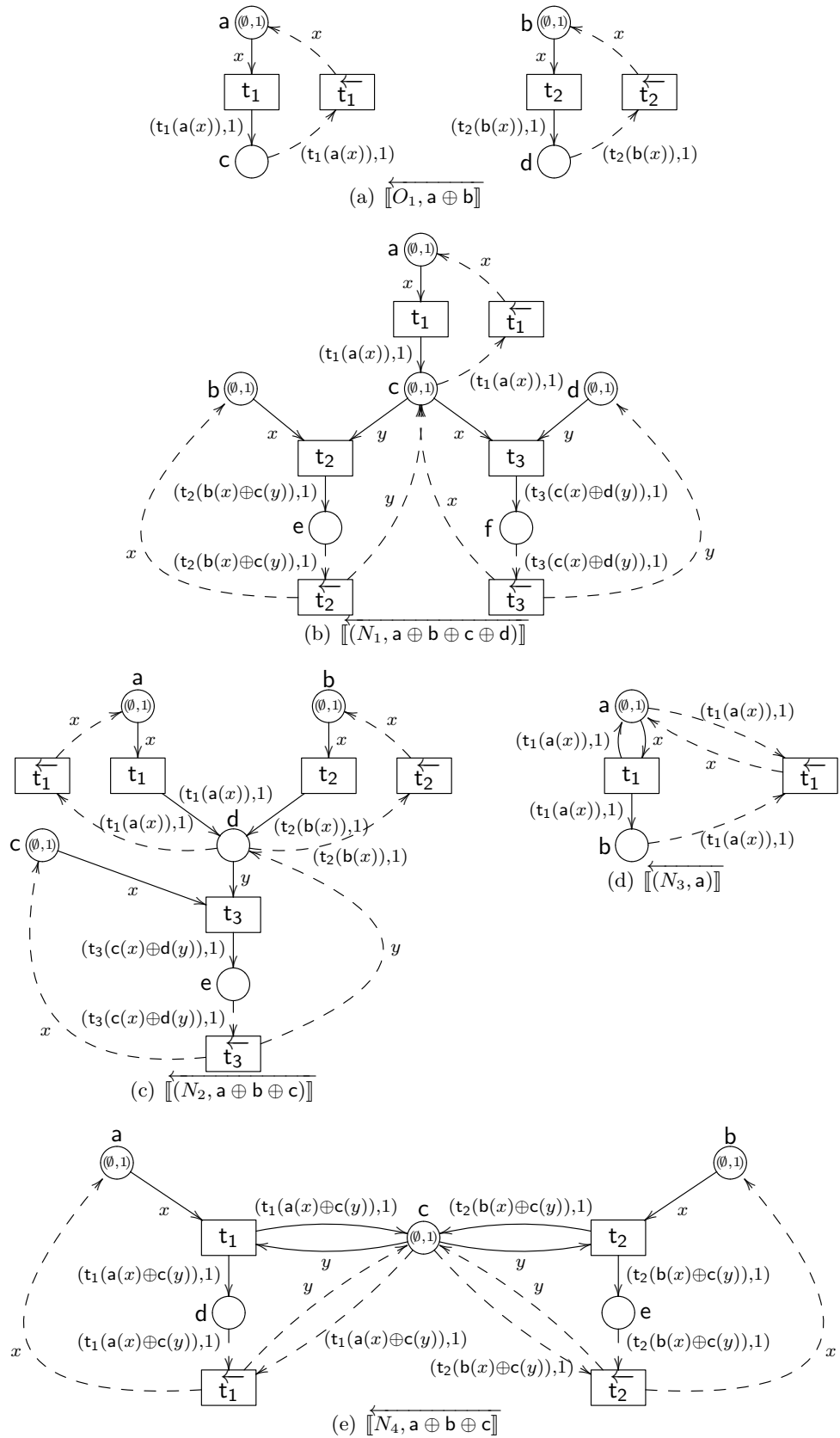

\begin{exa}
\label{ex:encoding-rnets}
Consider the {\sc p/t} net $(N_2,\p a \oplus \p b \oplus \p c )$ in
Figure~\ref{fig:N2}, whose reversible version is shown in
Figure~\ref{fig:rev-enc-N2}.
Such net is obtained by (i) mapping $(N_2,\p a \oplus \p b \oplus \p
c)$ into the {\sc c-p/t} net $\enc{(N_2,\p a \oplus \p b \oplus \p
c)}$ shown in Figure~\ref{fig:enc-N2}, and (ii) adding the reverse
transitions.
We now illustrate the execution of $ (\rev{\enc{N_2}}, \p
a(\emptyset,1) \oplus \p b(\emptyset,1) \oplus \p c(\emptyset,1))$.
Consider the following forward computation consisting of the firing of
$\tr{t_1}$ followed by $\tr{t_3}$.
{\small
\begin{align*}
  &
  (\rev{\enc{N_2}}, \p a(\emptyset,1)
  \oplus \p b(\emptyset,1)
  \oplus \p c(\emptyset,1))
  \\
  \red [{\tr {t_1}}] 
  &
  (\rev{\enc{N_2}}, \p b(\emptyset,1)
  \oplus \p c(\emptyset,1)
  \oplus \p d(\tr{t_1}(\p a(\emptyset,1)),1)  )
  \\
  \red[{\tr{t_3}}]
  &
  (\rev{\enc{N_2}},  \p b(\emptyset,1)
  \oplus \p e(\tr{t_3}(\p c(\emptyset,1)
  \oplus \p d(\tr{t_1}(\p a(\emptyset,1)),1) ,1)))
\end{align*}
}
Suppose, we would like to undo the above computation. The last marking
only enables ${\rev{\tr{t_3}}}$, which can be fired as follows
{\small
\begin{align*}
  (\rev{\enc{N_2}},  \p b(\emptyset,1)
  \oplus \p e(\tr{t_3}(\p c(\emptyset,1)
  \oplus \p d(\tr{t_1}(\p a(\emptyset,1)),1) ,1)))
  \bred[\rev{\tr{t_3}}]
  (\rev{\enc{N_2}}, \p b(\emptyset,1)
  \oplus \p c(\emptyset,1)
  \oplus \p d(\tr{t_1}(\p a(\emptyset,1)),1)  ).
\end{align*}
}
Note that the obtained marking enables, as expected, $\tr{t_3}$
again.
Additionally, it also enables $\rev{\tr{t_1}}$ because the
token in $\p d$ matches the preset of $\rev{\tr{t_1}}$. As expected,
the the firing of $\rev{\tr{t_1}}$ returns to the initial marking
{\small
\begin{align*}
  (\rev{\enc{N_2}}, \p b(\emptyset,1)
  \oplus \p c(\emptyset,1)
  \oplus \p d(\tr{t_1}(\p a(\emptyset,1)),1)  )
  \bred[\rev{\tr{t_1}}]
  (\rev{\enc{N_2}}, \p a(\emptyset,1)
  \oplus \p b(\emptyset,1)
  \oplus \p c(\emptyset,1)).
\end{align*}
}
Finally, we remark that the colour in the preset $\rev{\tr{t_2}}$
prevents the firing of $\rev{\tr{t_2}}$ in the marking of
$(\rev{\enc{N_2}}, \p b(\emptyset,1) \oplus \p
c(\emptyset,1) \oplus \p d(\tr{t_1}(\p a(\emptyset,1)),1) )$, because
the colour $\tr{t_1}(\p a(\emptyset,1))$ in the token in $\p d$ does
not match the one in the preset.
\end{exa}

\begin{exa}
We now illustrate the behaviour of $(\rev{\enc{N_3}},\p a)$ in
Figure~\ref{fig:rev-enc-N3} corresponding to the reversible version of
$(N_3,\p a)$ in Figure~\ref{fig:N3}. Consider the following forward
computation consisting of two consecutive firings of $\tr{t_1}$.
\[
\begin{array}{ll}
  &
  (\enc{\rev{N_3}},\p a(\emptyset,1)) 
  \\
  &
  \red [{\tr {t_1}}] 
  (\enc{\rev{N_3}},\p a(\tr {t_1}(\p a(\emptyset,1)),1)
  \oplus \p b(\tr {t_1}(\p a(\emptyset,1)),1))
  \\
  &
  \red [{\tr {t_1}}] 
  (\enc{\rev{N_3}},\p a(\tr {t_1}(\p a(\tr {t_1}(\p a(\emptyset,1)),1)),1)
  \oplus \p b(\tr {t_1}(\p a(\tr {t_1}(\p a(\emptyset,1)),1)),1)
  \oplus  \p b(\tr {t_1}(\p a(\emptyset,1)),1))
\end{array}
\]
Note that $\tr{t_1}$ can be fired infinitely many times. At any point,
the last firing can be undone by applying $\rev{\tr{t_1}}$. Also note
that the colours in the preset of $\rev{\tr{t_1}}$ ensure that the
last produced token in $\p b$, whose colour matches the one in $\p a$,
is consumed.
Then, we may have the following backward computation.
\[
\begin{array}{ll}
 &
 (\enc{\rev{N_3}},\p a(\tr {t_1}(\p a(\tr {t_1}(\p a(\emptyset,1)),1)),1)
 \oplus \p b(\tr {t_1}(\p a(\tr {t_1}(\p a(\emptyset,1)),1)),1)
 \oplus  \p b(\tr {t_1}(\p a(\emptyset,1)),1))
 \\
 \bred [\rev{\tr {t_1}}]
 &
 (\enc{\rev{N_3}},\p a(\tr {t_1}(\p a(\emptyset,1)),1)
 \oplus \p b(\tr {t_1}(\p a(\emptyset,1)),1))
 \\
 \bred[\rev{\tr{t_1}}]
 &
 (\enc{\rev{N_3}},\p a(\emptyset,1)) 
\end{array}
\]
\end{exa}

The following result states that the reductions of the
reversible \textsc{c-p/t} version of a net are in one-to-one
correspondence with the reductions of its reversible version of its
unfolding.
\begin{thm}[Correctness] 
\label{theorem:correspondence-unfolding-coloured}
Let $(N,m)$ be a marked \textsc{p/t} net and $\mathcal{U}\lbrack
N,m\rbrack = (O,m')$ its unfolding. Then, $(\rev{\enc N}, \enc
m) \red[s] (\rev{\enc{N}},m'')$ if and only if $(\rev O,m') \red[s]
(\rev O,m'')$.
\end{thm}
\begin{proof}
It follows by induction on the length of the reduction analogously to
the proof of Lemma~\ref{lemma:correspondence-unfolding-coloured}.
\end{proof}

\section{Conclusions}

This paper addresses the problem of reversing computations in Petri
nets.  The problem of reversing computation is different from the
classical property of reversibility in Petri nets~\cite{revPT}, which
refers to the ability of returning to its initial marking from any
reachable marking without requiring to traverse back the states of the
system.  In this sense, making a net reversible equates to adding a
minimal number of (forward) transitions that return us to the initial
marking~\cite{BarylskaKMP18,MikulskiL19}.  Reversibility is a global
property while reversing a computation is a local one, as discussed
in~\cite{BarylskaKMP18,MikulskiL19}.

Reversing computation in Petri nets has also been studied
in \cite{PhilippouP18}, where reversible Petri nets (RPNs) are
introduced.  RPNs are Petri nets endowed with a new kind of tokens,
called bonds, and a computational history that records the order of
execution. The motivation for bonds comes from bonds in biochemical
reactions. Transitions produce bonds that are constructed from the
consumed bonds.
For example, the bond generated by a transition that consumes a bond
from place $\p a$ and bond from place $\p b$ may have the shape $\p a
- \p b$.
In this way, each bond keeps track of its causal history and uses this
information for reversal.
Additionally, RPNs record the execution history in order to handle the
reversal of non-deterministic choices.
Informally, each firing in a computation is assigned with an integer
value, which induces an ordering on firings.  The main motivation to
have bonds is the possibility to model out-of-causal order
reversibility \cite{PhiUliYuen12,UK16,UK18}. Moreover, they are able
to model also causal-consistent
reversibility~\cite{rccs,ccsk,cc_bulletin} and
backtracking \cite{HoeyU19}.  Also, a translation of RPNs into
coloured Petri nets has been given \cite{BarylskaGMPPP18}.
Reversibility in RPN can be controlled by adding conditions on
backward transitions. In this way it is possible to model wireless
communications scenarios \cite{PhilippouPS19}.

We have presented a causally reversible semantics for P/T nets based
on two observations.  First, an occurrence net can be
straightforwardly reversed by adding a reverse version for each
(forward) transition.  Second, the standard unfolding construction
associates a P/T net with an occurrence net that preserves all of its
computation. Consequently, the reversible semantics of a P/T net can
be obtained as the reversible semantics of its unfolding.  We have
showed that reversibility in reversible occurrence net is
causal-consistent, namely that it preserves causality.  The unfolding
of an occurrence net can be infinite, for example when the original
P/T net is not acyclic.  Therefore we have shown that the reversible
behaviour of reversible occurrence nets can be expressed as a finite
net whose tokens are coloured by causal histories.
Colours in our encoding resemble the causal memories that are typical
in reversible process calculi~\cite{rhotcs,rccs}.  We plan to
implement an algorithm to automate our translation into coloured Petri
nets and to integrate it with the CPN tool \cite{cpntool}. This will
allow us to simulate our reversible nets and eventually could provide
visual analysis support for reversible
debuggers \cite{GiachinoLM14,Lanese0PV18,HoeyU19}.

Occurrence nets have a direct mapping into prime event structures. 
Following the research line undertaken in \cite{rpes_ron},
we will continue investigating in the future the relation between
reversible event
structures~\cite{PhillipsU15,CristescuKV15,UlidowskiPY18,GraversenPY18}
and our reversible occurrence nets.  There are alternative methods
for proving causal-consistent reversibility in a reversible model of
computation. They are based on showing other properties than those in
Section~4, either the Well-Foundedness (lack of infinite reverse
sequences) and Reverse Diamond properties in \cite{ccsk,PhiUli07} or
the properties in~\cite{revaxiom}. It would be worthwhile to prove the
alternative properties for our reversible nets, and compare the two
approaches.

\section*{Acknowledgment}
We would like to thank the anonymous reviewers for their helpful
suggestions of improvements.

\bibliographystyle{alpha}
\bibliography{biblio}

\newcommand{\etalchar}[1]{$^{#1}$}
\begin{thebibliography}{MMP{\etalchar{+}}20}

\bibitem[BGM{\etalchar{+}}18]{BarylskaGMPPP18}
Kamila Barylska, Anna Gogoli{\' n}ska, {\L}ukasz Mikulski, Anna Philippou,
  Marcin Pi{\c a}tkowski, and Kyriaki Psara.
\newblock Reversing computations modelled by coloured {P}etri nets.
\newblock In Wil M.~P. van~der Aalst, Robin Bergenthum, and Josep Carmona,
  editors, {\em Proceedings of the International Workshop on Algorithms {\&}
  Theories for the Analysis of Event Data}, volume 2115 of {\em {CEUR} Workshop
  Proceedings}, pages 91--111. CEUR-WS.org, 2018.

\bibitem[BKMP18]{BarylskaKMP18}
Kamila Barylska, Maciej Koutny, {\L}ukasz Mikulski, and Marcin Pi{\c a}tkowski.
\newblock Reversible computation vs. reversibility in {P}etri nets.
\newblock {\em Sci. Comput. Program.}, 151:48--60, 2018.

\bibitem[CKV13]{CristescuKV13}
Ioana Cristescu, Jean Krivine, and Daniele Varacca.
\newblock A compositional semantics for the reversible p-calculus.
\newblock In {\em 28th Annual {ACM/IEEE} Symposium on Logic in Computer
  Science, {LICS}}, pages 388--397. {IEEE} Computer Society, 2013.

\bibitem[CKV15]{CristescuKV15}
Ioana Cristescu, Jean Krivine, and Daniele Varacca.
\newblock Rigid families for {CCS} and the {\(\pi\)}-calculus.
\newblock In Martin Leucker, Camilo Rueda, and Frank~D. Valencia, editors, {\em
  Theoretical Aspects of Computing - {ICTAC} 2015 - 12th International
  Colloquium}, volume 9399 of {\em Lecture Notes in Computer Science}, pages
  223--240. Springer, 2015.

\bibitem[CL11]{CardelliL11}
Luca Cardelli and Cosimo Laneve.
\newblock Reversible structures.
\newblock In Fran{\c{c}}ois Fages, editor, {\em Computational Methods in
  Systems Biology, 9th International Conference, {CMSB} 2011}, pages 131--140.
  {ACM}, 2011.

\bibitem[CLM76]{revPT}
E.~Cardoza, Richard~J. Lipton, and Albert~R. Meyer.
\newblock Exponential space complete problems for {P}etri nets and commutative
  semigroups (preliminary report).
\newblock In {\em Proceedings of the Eighth Annual ACM Symposium on Theory of
  Computing}, STOC '76, pages 50--54, New York, NY, USA, 1976. ACM.

\bibitem[DK04]{rccs}
Vincent Danos and Jean Krivine.
\newblock Reversible communicating systems.
\newblock In Philippa Gardner and Nobuko Yoshida, editors, {\em {CONCUR} 2004 -
  Concurrency Theory, 15th International Conference}, volume 3170 of {\em
  Lecture Notes in Computer Science}, pages 292--307. Springer, 2004.

\bibitem[DK05]{DanosK05}
Vincent Danos and Jean Krivine.
\newblock Transactions in {RCCS}.
\newblock In Mart{\'{\i}}n Abadi and Luca de~Alfaro, editors, {\em {CONCUR}
  2005}, volume 3653 of {\em Lecture Notes in Computer Science}, pages
  398--412. Springer, 2005.

\bibitem[GLM14]{GiachinoLM14}
Elena Giachino, Ivan Lanese, and Claudio~Antares Mezzina.
\newblock Causal-consistent reversible debugging.
\newblock In Stefania Gnesi and Arend Rensink, editors, {\em Fundamental
  Approaches to Software Engineering - 17th International Conference, {FASE}
  2014}, volume 8411 of {\em Lecture Notes in Computer Science}, pages
  370--384. Springer, 2014.

\bibitem[GPY18]{GraversenPY18}
Eva Graversen, Iain Phillips, and Nobuko Yoshida.
\newblock Event structure semantics of (controlled) reversible {CCS}.
\newblock In Jarkko Kari and Irek Ulidowski, editors, {\em Reversible
  Computation - 10th International Conference}, volume 11106 of {\em Lecture
  Notes in Computer Science}, pages 102--122. Springer, 2018.

\bibitem[HLN{\etalchar{+}}20]{HoeyLNUV20}
James Hoey, Ivan Lanese, Naoki Nishida, Irek Ulidowski, and Germ{\'{a}}n Vidal.
\newblock A case study for reversible computing: Reversible debugging of
  concurrent programs.
\newblock In Irek Ulidowski, Ivan Lanese, Ulrik~Pagh Schultz, and Carla
  Ferreira, editors, {\em Reversible Computation: Extending Horizons of
  Computing - Selected Results of the {COST} Action {IC1405}}, volume 12070 of
  {\em Lecture Notes in Computer Science}, pages 108--127. Springer, 2020.

\bibitem[HU19]{HoeyU19}
James Hoey and Irek Ulidowski.
\newblock Reversible imperative parallel programs and debugging.
\newblock In Michael~Kirkedal Thomsen and Mathias Soeken, editors, {\em
  Reversible Computation - 11th International Conference, {RC} 2019}, volume
  11497 of {\em Lecture Notes in Computer Science}, pages 108--127. Springer,
  2019.

\bibitem[HUY18]{JH2018}
James Hoey, Irek Ulidowski, and Shoji Yuen.
\newblock Reversing parallel programs with blocks and procedures.
\newblock In Jorge~A. P{\'{e}}rez and Simone Tini, editors, {\em Proceedings
  Combined 25th International Workshop on Expressiveness in Concurrency and
  15th Workshop on Structural Operational Semantics, {EXPRESS/SOS} 2018},
  volume 276 of {\em {EPTCS}}, pages 69--86, 2018.

\bibitem[HW08]{HaymanW08}
Jonathan Hayman and Glynn Winskel.
\newblock The unfolding of general {P}etri nets.
\newblock In Ramesh Hariharan, Madhavan Mukund, and V.~Vinay, editors, {\em
  {IARCS} Annual Conference on Foundations of Software Technology and
  Theoretical Computer Science, {FSTTCS} 2008, December 9-11, 2008, Bangalore,
  India}, volume~2 of {\em LIPIcs}, pages 223--234. Schloss Dagstuhl -
  Leibniz-Zentrum fuer Informatik, 2008.

\bibitem[JKW07]{cpntool}
Kurt Jensen, Lars~Michael Kristensen, and Lisa Wells.
\newblock Coloured {P}etri nets and {CPN} tools for modelling and validation of
  concurrent systems.
\newblock {\em {STTT}}, 9(3-4):213--254, 2007.

\bibitem[KAC{\etalchar{+}}20]{KACPPU20}
Stefan Kuhn, Bogdan Aman, Gabriel Ciobanu, Anna Philippou, Kyriaki Psara, and
  Irek Ulidowski.
\newblock Reversibility in chemical reactions.
\newblock In Irek Ulidowski, Ivan Lanese, Ulrik~Pagh Schultz, and Carla
  Ferreira, editors, {\em Reversible Computation: Extending Horizons of
  Computing - Selected Results of the {COST} Action {IC1405}}, volume 12070 of
  {\em Lecture Notes in Computer Science}, pages 151--176. Springer, 2020.

\bibitem[KU16]{UK16}
Stefan Kuhn and Irek Ulidowski.
\newblock A calculus for local reversibility.
\newblock In Simon~J. Devitt and Ivan Lanese, editors, {\em Reversible
  Computation - 8th International Conference, {RC} 2016}, volume 9720 of {\em
  Lecture Notes in Computer Science}, pages 20--35. Springer, 2016.

\bibitem[KU18]{UK18}
Stefan Kuhn and Irek Ulidowski.
\newblock Local reversibility in a {C}alculus of {C}ovalent {B}onding.
\newblock {\em Sci. Comput. Program.}, 151:18--47, 2018.

\bibitem[Lee86]{Leeman}
George~B. Leeman, Jr.
\newblock A formal approach to undo operations in programming languages.
\newblock {\em ACM Trans. Program. Lang. Syst.}, 8(1):50--87, January 1986.

\bibitem[L{\'{e}}v76]{levy}
Jean{-}Jacques L{\'{e}}vy.
\newblock An algebraic interpretation of the \emph{lambda beta} k-calculus; and
  an application of a labelled \emph{lambda}-calculus.
\newblock {\em Theor. Comput. Sci.}, 2(1):97--114, 1976.

\bibitem[LLM{\etalchar{+}}13]{LaneseLMSS13}
Ivan Lanese, Michael Lienhardt, Claudio~Antares Mezzina, Alan Schmitt, and
  Jean{-}Bernard Stefani.
\newblock Concurrent flexible reversibility.
\newblock In Matthias Felleisen and Philippa Gardner, editors, {\em Programming
  Languages and Systems - 22nd European Symposium on Programming, {ESOP} 2013},
  volume 7792 of {\em Lecture Notes in Computer Science}, pages 370--390.
  Springer, 2013.

\bibitem[LMS16]{rhotcs}
Ivan Lanese, Claudio~Antares Mezzina, and Jean{-}Bernard Stefani.
\newblock Reversibility in the higher-order {\(\pi\)}-calculus.
\newblock {\em Theor. Comput. Sci.}, 625:25--84, 2016.

\bibitem[LMSS11]{LaneseMSS11}
Ivan Lanese, Claudio~Antares Mezzina, Alan Schmitt, and Jean{-}Bernard Stefani.
\newblock Controlling reversibility in higher-order pi.
\newblock In Joost{-}Pieter Katoen and Barbara K{\"{o}}nig, editors, {\em
  {CONCUR} 2011 - Concurrency Theory - 22nd International Conference}, volume
  6901 of {\em Lecture Notes in Computer Science}, pages 297--311. Springer,
  2011.

\bibitem[LMT14]{cc_bulletin}
Ivan Lanese, Claudio~Antares Mezzina, and Francesco Tiezzi.
\newblock Causal-consistent reversibility.
\newblock {\em Bulletin of the {EATCS}}, 114, 2014.

\bibitem[LNPV18]{Lanese0PV18}
Ivan Lanese, Naoki Nishida, Adri{\'{a}}n Palacios, and Germ{\'{a}}n Vidal.
\newblock Cauder: {A} causal-consistent reversible debugger for {E}rlang.
\newblock In John~P. Gallagher and Martin Sulzmann, editors, {\em Functional
  and Logic Programming - 14th International Symposium, {FLOPS}}, volume 10818
  of {\em Lecture Notes in Computer Science}, pages 247--263. Springer, 2018.

\bibitem[LPU20]{revaxiom}
Ivan Lanese, Iain C.~C. Phillips, and Irek Ulidowski.
\newblock An axiomatic approach to reversible computation.
\newblock In Jean Goubault{-}Larrecq and Barbara K{\"{o}}nig, editors, {\em
  Foundations of Software Science and Computation Structures - 23rd
  International Conference, {FOSSACS} 2020}, volume 12077 of {\em Lecture Notes
  in Computer Science}, pages 442--461. Springer, 2020.

\bibitem[ML19]{MikulskiL19}
{\L}ukasz Mikulski and Ivan Lanese.
\newblock Reversing unbounded {P}etri nets.
\newblock In Susanna Donatelli and Stefan Haar, editors, {\em Application and
  Theory of {P}etri Nets and Concurrency - 40th International Conference,
  {{P}etri} {NETS} 2019}, volume 11522 of {\em Lecture Notes in Computer
  Science}, pages 213--233. Springer, 2019.

\bibitem[MMP{\etalchar{+}}20]{rpes_ron}
Hern{\'{a}}n~C. Melgratti, Claudio~Antares Mezzina, Iain Phillips, G.~Michele
  Pinna, and Irek Ulidowski.
\newblock Reversible occurrence nets and causal reversible prime event
  structures.
\newblock In Ivan Lanese and Mariusz Rawski, editors, {\em Reversible
  Computation - 12th International Conference, {RC} 2020}, volume 12227 of {\em
  Lecture Notes in Computer Science}, pages 35--53. Springer, 2020.

\bibitem[MMPY18]{family_pi}
Doriana Medic, Claudio~Antares Mezzina, Iain Phillips, and Nobuko Yoshida.
\newblock A parametric framework for reversible pi-calculi.
\newblock In Jorge~A. P{\'{e}}rez and Simone Tini, editors, {\em Proceedings
  Combined 25th International Workshop on Expressiveness in Concurrency and
  15th Workshop on Structural Operational Semantics and 15th Workshop on
  Structural Operational Semantics, {EXPRESS/SOS}}, volume 276 of {\em
  {EPTCS}}, pages 87--103, 2018.

\bibitem[MMPY20]{stm}
Doriana Medic, Claudio~Antares Mezzina, Iain Phillips, and Nobuko Yoshida.
\newblock Towards a formal account for software transactional memory.
\newblock In Ivan Lanese and Mariusz Rawski, editors, {\em Reversible
  Computation - 12th International Conference, {RC} 2020}, volume 12227 of {\em
  Lecture Notes in Computer Science}, pages 255--263. Springer, 2020.

\bibitem[MMU19]{MelgrattiMU19}
Hern{\'{a}}n~C. Melgratti, Claudio~Antares Mezzina, and Irek Ulidowski.
\newblock Reversing {P/T} nets.
\newblock In Hanne~Riis Nielson and Emilio Tuosto, editors, {\em Coordination
  Models and Languages - 21st {IFIP} {WG} 6.1 International Conference,
  {COORDINATION} 2019, Held as Part of the 14th International Federated
  Conference on Distributed Computing Techniques, DisCoTec 2019}, volume 11533
  of {\em Lecture Notes in Computer Science}, pages 19--36. Springer, 2019.

\bibitem[NPW81]{NielsenPW81}
Mogens Nielsen, Gordon~D. Plotkin, and Glynn Winskel.
\newblock {P}etri nets, event structures and domains, part {I}.
\newblock {\em Theor. Comput. Sci.}, 13:85--108, 1981.

\bibitem[Pin17]{Pinna17}
G.~Michele Pinna.
\newblock Reversing steps in membrane systems computations.
\newblock In Marian Gheorghe, Grzegorz Rozenberg, Arto Salomaa, and Claudio
  Zandron, editors, {\em Membrane Computing - 18th International Conference,
  {CMC} 2017}, volume 10725 of {\em Lecture Notes in Computer Science}, pages
  245--261. Springer, 2017.

\bibitem[PP18]{PhilippouP18}
Anna Philippou and Kyriaki Psara.
\newblock Reversible computation in {P}etri nets.
\newblock In Jarkko Kari and Irek Ulidowski, editors, {\em Reversible
  Computation - 10th International Conference, {RC} 2018}, volume 11106 of {\em
  Lecture Notes in Computer Science}, pages 84--101. Springer, 2018.

\bibitem[PPv19]{PhilippouPS19}
Anna Philippou, Kyriaki Psara, and Harun \v{S}iljak.
\newblock Controlling reversibility in reversing {P}etri nets with application
  to wireless communications - work-in-progress paper.
\newblock In Michael~Kirkedal Thomsen and Mathias Soeken, editors, {\em
  Reversible Computation - 11th International Conference, {RC} 2019}, volume
  11497 of {\em Lecture Notes in Computer Science}, pages 238--245. Springer,
  2019.

\bibitem[PU07a]{PhiUli07}
Iain Phillips and Irek Ulidowski.
\newblock Reversibility and models for concurrency.
\newblock In {\em Proceedings of 4th Workshop on Structural Operational
  Semantics SOS 2007}, volume 192 of {\em ENTCS}, pages 93--108, 2007.

\bibitem[PU07b]{ccsk}
Iain Phillips and Irek Ulidowski.
\newblock Reversing algebraic process calculi.
\newblock {\em J. Log. Algebr. Program.}, 73(1-2):70--96, 2007.

\bibitem[PU15]{PhillipsU15}
Iain Phillips and Irek Ulidowski.
\newblock Reversibility and asymmetric conflict in event structures.
\newblock {\em J. Log. Algebr. Meth. Program.}, 84(6):781--805, 2015.

\bibitem[PUY13]{PhiUliYuen12}
Iain Phillips, Irek Ulidowski, and Shoji Yuen.
\newblock A reversible process calculus and the modelling of the {ERK}
  signalling pathway.
\newblock In {\em Proceedings of Reversible Computation, {RC} 2012}, volume
  7581 of {\em Lecture Notes in Computer Science}, pages 218--232. Springer,
  2013.

\bibitem[SOJJ18]{SchordanOJB18}
Markus Schordan, Tomas Oppelstrup, David~R. Jefferson, and Peter D.~Barnes Jr.
\newblock Generation of reversible {C++} code for optimistic parallel discrete
  event simulation.
\newblock {\em New Generation Comput.}, 36(3):257--280, 2018.

\bibitem[UPY18]{UlidowskiPY18}
Irek Ulidowski, Iain Phillips, and Shoji Yuen.
\newblock Reversing event structures.
\newblock {\em New Generation Comput.}, 36(3):281--306, 2018.

\bibitem[VBR18]{2018Vos}
Alexis~De Vos, Stijn~De Baerdemacker, and Yvan~Van Rentergem.
\newblock {\em Synthesis of Quantum Circuits vs. Synthesis of Classical
  Reversible Circuits}.
\newblock Synthesis Lectures on Digital Circuits and Systems. Morgan {\&}
  Claypool Publishers, 2018.

\bibitem[VS18]{VassorS18}
Martin Vassor and Jean{-}Bernard Stefani.
\newblock Checkpoint/rollback vs causally-consistent reversibility.
\newblock In Jarkko Kari and Irek Ulidowski, editors, {\em Reversible
  Computation - 10th International Conference, {RC} 2018}, volume 11106 of {\em
  Lecture Notes in Computer Science}, pages 286--303. Springer, 2018.

\bibitem[Win86]{Winskel86}
Glynn Winskel.
\newblock Event structures.
\newblock In {\em {P}etri Nets: Central Models and Their Properties, Advances
  in {P}etri Nets 1986, Part II, Proceedings of an Advanced Course, Bad Honnef,
  Germany, 8-19 September 1986}, pages 325--392, 1986.

\end{thebibliography}

\end{document}